\documentclass[onecolumn,draftclsnofoot,12pt]{IEEEtran}
\usepackage{amsfonts}
\usepackage{times}
\usepackage{graphicx}
\usepackage{latexsym}
\usepackage{dsfont}
\usepackage{amssymb}
\usepackage{amsmath}
\usepackage{cite}
\usepackage{verbatim}
\usepackage{subfigure}
\newtheorem{theorem}{Theorem}

\newtheorem{corollary}[theorem]{Corollary}

\newtheorem{lemma}[theorem]{Lemma}

\newenvironment{proof}{ \textbf{Proof:} }{ \hfill $\Box$}

\newcommand{\figref}[1]{{Fig.}~\ref{#1}}

\def\bb0{{\mathbb{0}}}

\def\ba{{\mathbf{a}}}
\def\bb{{\mathbf{b}}}

\def\bh{{\mathbf{h}}}

\def\bm{{\mathbf{m}}}

\def\bs{{\mathbf{s}}}

\def\bu{{\mathbf{u}}}

\def\bw{{\mathbf{w}}}

\def\b0{{\mathbf{0}}}

\def\bA{{\mathbf{A}}}
\def\bB{{\mathbf{B}}}

\def\bF{{\mathbf{F}}}
\def\bG{{\mathbf{G}}}
\def\bH{{\mathbf{H}}}
\def\bI{{\mathbf{I}}}

\def\bR{{\mathbf{R}}}

\def\bU{{\mathbf{U}}}

\def\bW{{\mathbf{W}}}

\def\bbE{{\mathbb{E}}}

\def\cA{\mathcal{A}}

\def\cK{\mathcal{K}}

\def\cN{\mathcal{N}}

\def\sf0{{\mathsf{0}}}

\allowdisplaybreaks

\usepackage{epstopdf}
\usepackage{enumerate}
\usepackage{algorithmicx}
\usepackage{algorithm}
\usepackage{amsmath}
\usepackage[noend]{algpseudocode}
\usepackage{float}
\usepackage{color}
\usepackage{makeidx}
\usepackage{bbm}

\newcommand{\sref}[1]{{Section}~\ref{#1}}

\newcommand{\pinv}[1]{\ensuremath{#1^{\dagger}}} 	

\def \bUpsilon {\boldsymbol{\Upsilon}}
\def \bUN {\bU^\mathrm{NI}}
\def \rm {\mathrm}
\def \rA {r_\mathrm{A}}
\def \rE {r_\mathrm{E}}
\def \rN {r_\mathrm{NI}}
\begin{document}
\title{ Multi-Layer Precoding: A Potential Solution for Full-Dimensional Massive MIMO Systems}
\author{Ahmed Alkhateeb, Geert Leus, and Robert W. Heath, Jr. \thanks{Ahmed Alkhateeb and Robert W. Heath Jr. are with The University of Texas at Austin (Email: aalkhateeb, rheath@utexas.edu). Geert Leus is with Delft University of Technology (Email: g.j.t.leus@tudelft.nl).} \thanks{This material is based upon work supported in part by the National Science Foundation under Grant No. NSF-CCF-1319556.}\thanks{A shorter version of this paper was presented at Asilomar Conference on Signals, Systems, and Computers, Nov., 2014 \cite{Alkhateeb2014c}.}}
\maketitle

\begin{abstract}
Massive multiple-input multiple-output (MIMO) systems achieve high sum spectral efficiency by offering an order of magnitude increase in multiplexing gains. In time division duplexing systems, however, the reuse of uplink training pilots among cells results in additional channel estimation error, which causes downlink inter-cell interference, even when large numbers of antennas are employed. Handling this interference with conventional network MIMO techniques is challenging due to the large channel dimensionality. Further, the implementation of large antenna precoding/combining matrices is associated with high hardware complexity and power consumption. In this paper, we propose multi-layer precoding to enable efficient and low complexity operation in full-dimensional massive MIMO, where a large number of antennas is used in two dimensions. In multi-layer precoding, the precoding matrix of each base station is written as a product of a number of precoding matrices, each one called a layer. Multi-layer precoding (i) leverages the directional characteristics of large-scale MIMO channels to manage inter-cell interference with low channel knowledge requirements, and (ii) allows for an efficient implementation using low-complexity hybrid analog/digital architectures. We present a specific multi-layer precoding design for full-dimensional massive MIMO systems. The performance of this precoding design is analyzed and the per-user achievable rate is characterized for general channel models. The asymptotic optimality of the proposed multi-layer precoding design is then proved for some special yet important channels. Numerical simulations  verify the analytical results and illustrate the potential gains of multi-layer precoding compared to traditional pilot-contaminated massive MIMO solutions.
\end{abstract}

\section{Introduction} \label{sec:Intro}
 
Massive MIMO promises significant spectral efficiency gains for cellular systems. Scaling up the number of antennas, however, faces several challenges that prevent the corresponding scaling of the gains \cite{Larsson2014,Rusek2013,HeathJr2016,Truong2013}. First, the training and feedback of the large channels has high overhead in   frequency division duplexing (FDD) systems. To overcome that, channel reciprocity in conjunction with time division duplexing (TDD) systems is used \cite{Marzetta2010,Bjoernson2016}. Reusing the uplink training pilots among cells, however, causes channel estimation errors which in turn lead to downlink inter-cell interference, especially for cell-edge users \cite{Marzetta2010}. Managing this inter-cell interference using traditional network MIMO techniques requires high coordination overhead, which could limit the overall system performance \cite{Lozano2013}. Another challenge with the large number of antennas lies in the hardware implementation \cite{HeathJr2016,Singh2009}. Traditional MIMO precoding techniques normally assumes complete baseband processing, which requires dedicating an RF chain per antenna. This may lead to high cost and power consumption in massive MIMO systems \cite{HeathJr2016}. Therefore, developing precoding techniques that can overcome the challenges of inter-cell interference and complete baseband processing is of great interest.

\subsection{Prior Work}

Inter-cell interference is a critical problem for general MIMO systems. Typical solutions for managing this interference require some sort of collaboration between the base stations (BSs) \cite{Gesbert2010}. The overhead of this cooperation, though, can limit the system performance \cite{Lozano2013}. When the number of antennas grows to infinity, the performance of the network becomes limited by pilot contamination \cite{Marzetta2010}, which is one form of inter-cell interference. Pilot contamination happens because of  the channel estimation errors that result from reusing the uplink training pilots among users in TDD massive MIMO systems. Several solutions have been proposed to manage inter-cell interference in massive MIMO systems \cite{Huh2012,Jose2011,Ashikhmin2012,Yin2013}. In \cite{Huh2012,Jose2011}, multi-cell zero-forcing and MMSE MIMO precoding strategies were developed to cancel or reduce the inter-cell interference. The solutions in \cite{Huh2012,Jose2011}, however, require global channel knowledge at every BS, which makes them feasible only for small numbers of antennas \cite{Lu2014}. Pilot contamination precoding was proposed in \cite{Ashikhmin2012} to overcome the pilot contamination problem, relying on the channel covariance knowledge. The technique in \cite{Ashikhmin2012}, though, requires sharing the transmitted messages between all BSs, which is difficult to achieve in practice. In \cite{Yin2013}, the directional characteristics of large-dimensional channels were leveraged to improve the uplink channel training in TDD systems. This solution, however, requires fully-digital hardware and does not leverage the higher degrees of freedom provided in full-dimensional massive MIMO systems.

Precoding approaches that divide the processing between two stages have been developed in \cite{ElAyach2014,Alkhateeb2014b,Bogale2014,Liang2014,Adhikary2013} for mmWave and massive MIMO systems. Motivated by the high cost and power consumption of the RF, \cite{ElAyach2014} developed hybrid analog/digital precoding algorithms for mmWave systems. Hybrid precoding divides the precoding between RF and baseband domains, and requires a much smaller number of RF chains compared to the number of antennas. For multi-user systems \cite{Alkhateeb2014b} proposed a two-stage hybrid precoding design where the first precoding matrix is designed to maximize the signal power for each user and the second matrix is designed to manage the multi-user interference. Similar solutions were also developed for massive MIMO systems \cite{Bogale2014,Liang2014}, with the general objective of maximizing the system sum-rate. In \cite{Adhikary2013}, a two-stage joint spatial division and multiplexing (JSDM) precoding scheme was developed to reduce the channel training overhead in FDD massive MIMO systems. In JSDM, the base station (BS) divides the mobile stations (MSs) into groups of approximately similar covariance eigenspaces, and designs a pre-beamforming matrix based on the large channel statistics. The interference between the users of each group is then managed using another precoding matrix given the effective reduced-dimension channels. The work in \cite{ElAyach2014,Alkhateeb2014b,Bogale2014,Liang2014,Adhikary2013}, however, did not consider out-of-cell interference, which ultimately limits the performance of massive MIMO systems.

\subsection{Contribution} 
In this paper, we introduce a general framework, called multi-layer precoding, that (i) coordinates inter-cell interference in full-dimensional massive MIMO systems leveraging large channel characteristics and (ii) allows for efficient implementations using hybrid analog/digital architectures. Note that most of the literature on full-dimensional MIMO systems did not assume massive MIMO \cite{Nam2013,Kim2014a,Seifi2014a}, and the two systems were studied independently using different tools and theories. In this paper, we refer to full-dimensional massive MIMO as a two-dimensional MIMO system, which adopts large numbers of antennas in the two dimensions. The main contributions of our work are summarized as follows.
\begin{itemize}
	\item Designing a specific multi-layer precoding solution for full-dimensional massive MIMO systems. The proposed precoding strategy decouples the precoding matrix of each BS as a multiplication of three precoding matrices, called layers. The three precoding layers are designed to avoid inter-cell interference, maximize effective signal power, and manage intra-cell multi-user interference, with low channel training overhead. 
	\item Analyzing the performance of the proposed multi-layer precoding design. First, the per-user achievable rate using multi-layer precoding is derived for a general channel model. Then, asymptotic optimality results for the achievable rates with multi-layer precoding are derived for two special channel models: the one-ring and the single-path models. Lower bounds on the achievable rates for the cell-edge users are also characterized under the one-ring channel model.
\end{itemize}
The developed multi-layer precoding solutions are also evaluated by numerical simulations. Results show the multi-layer precoding can approach the single-user rate, which is free of inter-cell and intra-cell interference, in some special cases. Further, results illustrate that significant rate and coverage gains can be obtained by multi-layer precoding  compared to conventional conjugate beamforming and zero-forcing massive MIMO solutions.

We use the following notation throughout this paper: $\bA$ is a matrix, $\ba$ is a vector, $a$ is a scalar, and $\cA$ is a set. $|\bA|$ is the determinant of $\bA$, $\|\bA \|_F$ is its Frobenius norm, whereas $\bA^T$, $\bA^H$, $\bA^*$, $\bA^{-1}$, $\pinv{\bA}$ are its transpose, Hermitian (conjugate transpose), conjugate, inverse, and pseudo-inverse respectively. $[\bA]_{r,:}$ and $[\bA]_{:,c}$ are the $r$th row and $c$th column of the matrix $\bA$, respectively. $\mathrm{diag}(\ba)$ is a diagonal matrix with the entries of $\ba$ on its diagonal. $\bI$ is the identity matrix and $\mathbf{1}_{N}$ is the $N$-dimensional all-ones vector. $\bA \otimes \bB$ is the Kronecker product of $\bA$ and $\bB$, and $\bA \circ \bB$ is their Khatri-Rao product. $\cN(\bm,\bR)$ is a complex Gaussian random vector with mean $\bm$ and covariance $\bR$. $\bbE\left[\cdot\right]$ is used to denote expectation.

\section{System and Channel Models} \label{sec:Model}

In this section, we present the full-dimensional massive MIMO system and channel models adopted in the paper.

\subsection{System Model} \label{sec:SysModel}
\begin{figure}[t]
	\centerline{
		\includegraphics[width=.95\columnwidth]{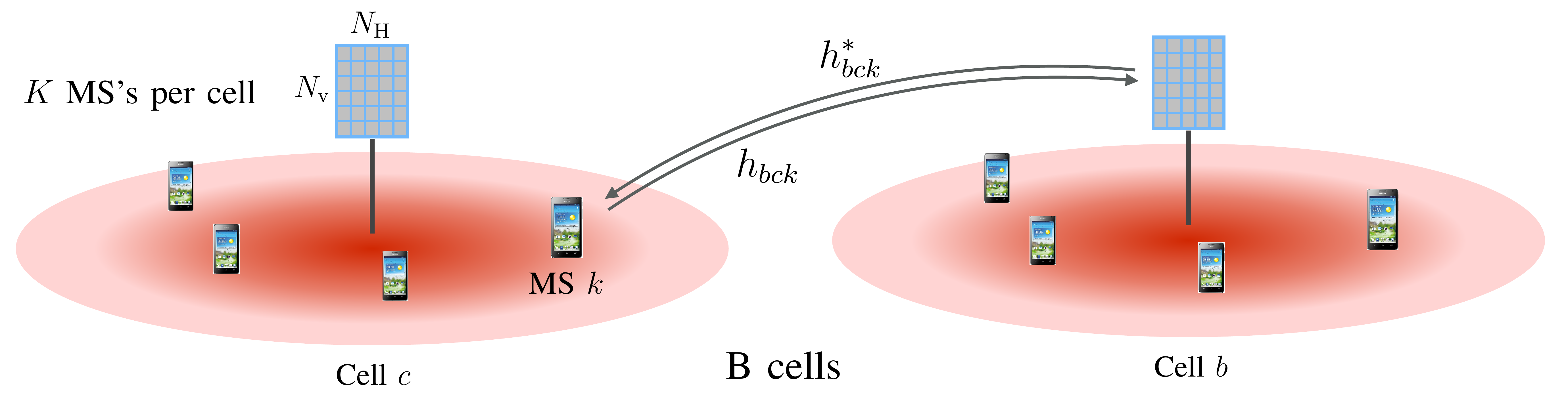}
	}
	\caption{A full-dimensional MIMO cellular model where each BS has a 2D antenna array and serves $K$ users.}
	\label{fig:Model}
\end{figure}

Consider a cellular system model consisting of $B$ cells with one BS and $K$ MS's in each cell, as shown in \figref{fig:Model}. Each BS is equipped with a two-dimensional (2D) antenna array of $N$ elements, $N=N_\mathrm{V}$ (vertical antennas) $\times N_\mathrm{H}$ (horizontal antennas), and each MS has a single antenna. We assume that all BSs and MSs are synchronized and operate a TDD protocol with universal frequency reuse. In the downlink, each BS $b=1, 2, ..., B$, applies an $N \times K$ precoder $\bF_b$ to transmit a symbol for each user, with a power constraint $\|\left[\bF_b\right]_{:,k}\|^2=1$, $k=1, 2, ..., K$. Uplink and downlink channels are assumed to be reciprocal. If $\bh_{b c k}$ denotes the $N \times 1$ uplink channel from  user $k$ in cell $c$ to BS $b$, then the received signal by this user in the downlink can be written as
\begin{equation}
y_{c k}=\sum_{b =1}^{B} \bh_{b c k}^* \bF_{b} \bs_{b} + n_{c k},
\label{eq:Received}
\end{equation}
where $\bs_{b}=\left[s_{b,1}, ..., s_{b,K}\right]^T$ is the $K \times 1$ vector of transmitted symbols from BS $b$, such that $\bbE \left[\bs_{b} \bs_{b}^*\right]=\frac{P}{K} \bI$, with $P$ representing the average total transmitted power, and $n_{c k}\sim \cN (0, \sigma^2 )$ is the Gaussian noise at user $k$ in cell $c$. It is useful to expand \eqref{eq:Received} as
\begin{equation}
y_{c k}=\underbrace{\bh_{c c k}^* \left[\bF_{c}\right]_{:,k} s_{c,k}}_\text{Desired signal}+\underbrace{ \sum_{m \neq k} \bh_{c c k}^* \left[\bF_{c}\right]_{:,m} s_{c,m}}_\text{Intra-cell interference}+ \underbrace{\sum_{b \neq c} \bh_{b c k}^* \bF_{b} \bs_{b}}_\text{Inter-cell interference} + n_{c k},
\end{equation}
to illustrate the different components of the received signal.

\subsection{Channel Model} \label{sec:ChModel}
We consider a full-dimensional MIMO configuration where 2D antenna arrays are deployed at the BS's. Consequently, the channels from the BS's to each user have a 3D structure. Extensive efforts are currently given to 3D channel measurements, modeling, and standardization \cite{Kammoun2014,Zhong2013}. One candidate is the Kronecker product correlation model, which provides a reasonable approximation of 3D covariance matrices\cite{Ying2014}. In this model, the covariance of the 3D channel $\bh_{b c k}$, defined as $\bR_{b c k}= \bbE\left[\bh_{b c k} \bh_{b c k}^*\right]$, is approximated by
\begin{equation}
\bR_{b c k}= \bR_{b c k}^\mathrm{A} \otimes \bR_{b c k}^\mathrm{E},
\end{equation}
where $\bR_{b c k}^\mathrm{A}$ and $\bR_{b c k}^\mathrm{E}$ represent the covariance matrices in the azimuth and elevation directions. If $\bR_{b c k}^\mathrm{A}= \bU_{b c k}^\mathrm{A} {\boldsymbol{\Lambda}_{b c k}^{\mathrm{A}}} {\bU_{b c k}^\mathrm{A}}^*$ and $\bR_{b c k}^\mathrm{E}= \bU_{b c k}^\mathrm{E} {\boldsymbol{\Lambda}_{b c k}^{\mathrm{E}}} {\bU_{b c k}^\mathrm{E}}^*$ are the eigenvalue decompositions of $\bR_{b c k}^\mathrm{A}$ and $\bR_{b c k}^\mathrm{E}$, then using Karhunen-Loeve representation, the channel $\bh_{b c k}$ can be expressed as 
\begin{equation} \label{eq:channel}
\bh_{b c k}= \left[\bU_{b c k}^\mathrm{A} {\boldsymbol{\Lambda}_{b c k}^{\mathrm{A}}}^{\frac{1}{2}} \otimes \bU_{b c k}^\mathrm{E} {\boldsymbol{\Lambda}_{b c k}^{\mathrm{E}}}^{\frac{1}{2}}\right] \bw_{b c k},
\end{equation}
where $\bw_{b c k} \sim \cN(\boldsymbol{0}, \bI)$ is a $\mathrm{rank}\left(\bR_{b c k}^\mathrm{A}\right) \mathrm{rank}\left(\bR_{b c k}^\mathrm{E}\right) \times 1$ \ vector, with $\mathrm{rank}(\bA)$ representing the rank of the matrix $\bA$. Without loss of generality, and to simplify the notation, we assume that all the users have the same ranks for the azimuth and elevation covariance matrices, which are denoted as $r_\mathrm{A}$ and $r_\mathrm{E}$.

\section{Multi-Layer Precoding: The General Concept} \label{sec:Concept}
In this section, we briefly introduce the motivation and general concept of multi-layer precoding. Given the system model in \sref{sec:Model}, the signal-to-interference-plus-noise ratio (SINR) at user $k$ in cell $c$ is
\begin{equation}
\mathrm{SINR}_{c k}=\frac{\frac{P}{K} \left|\bh_{c c k}^* \left[\bF_c\right]_{:,k}\right|^2}{\frac{P}{K} 
	\displaystyle{\sum_{m \neq k}} |\bh_{c c k}^* \left[\bF_c\right]_{:,m}|^2 
	+  \frac{P}{K} \displaystyle{\sum_{b \neq c}} \|\bh_{b c k}^* \bF_{b}\|^2 + \sigma^2},
\end{equation}
where $\left|\bh_{c c k}^* \left[\bF_c\right]_{:,k}\right|^2$ is the desired signal power, ${\sum_{m \neq k}} |\bh_{c c k}^* \left[\bF_c\right]_{:,m}|^2$ is the intra-cell multi-user interference, and ${\sum_{b \neq c}} \|\bh_{b c k}^* \bF_{b}\|^2$ is the inter-cell interference. Designing one precoding matrix per BS to manage all these kinds of signals by, for example, maximizing the system sum-rate is non-trivial. This normally leads to a non-convex problem whose closed-form solution is unknown \cite{Gesbert2010}. Also, coordinating inter-cell interference between BS's typically results in high cooperation overhead that makes the value of this cooperation limited \cite{Lozano2013}. Another challenge lies in the entire baseband implementation of these precoding matrices, which may yield high cost and power consumption in massive MIMO systems \cite{HeathJr2016}.  

Our objective is to design the precoding matrices, $\bF_{b}$, $b=1, 2, ..., B$, such that (i) they manage the inter-cell and intra-cell interference with low requirements on the channel knowledge, and (ii) they can be implemented using low-complexity hybrid analog/digital architectures \cite{Alkhateeb2014b}, i.e., with a small number of RF chains. Next, we present the main idea of multi-layer precoding, a potential solution to achieve these objectives.

Inspired by prior work on multi-user hybrid precoding \cite{Alkhateeb2014b} and joint spatial division multiplexing \cite{Adhikary2013}, and leveraging the directional characteristics of large-scale MIMO channels \cite{Yin2013}, we propose to design the precoding matrix $\bF_{c}$ as a product of a number of precoding matrices (layers). In this paper, we will consider a 3-layer precoding matrix
\begin{equation}
\bF_{c}= \bF_{c}^{(1)} \bF_{c}^{(2)} \bF_{c}^{(3)},
\end{equation}
where these precoding layers are designed according to the following criteria. 
\begin{itemize}
\item \textbf{One precoding objective per layer:} Each layer is designed to achieve only one precoding objective, e.g., maximizing desired signal power, minimizing inter-cell interference, or minimizing multi-user interference. This simplifies the precoding design problem and divides it into easier and/or convex sub-problems. Further, this decouples the required channel knowledge for each layer.

\item \textbf{Successive dimensionality reduction:} Each layer is designed such that the effective channel, including this layer, has smaller dimensions compared to the original channel. This reduces the channel training overhead of every precoding layer compared to the previous one. Further, this makes a successive reduction in the dimensions of the precoding matrices, which eases implementing them using hybrid analog/digital architectures \cite{HeathJr2016,ElAyach2014,Alkhateeb2014b,Han2015} with small number of RF chains. 

\item \textbf{Different channel statistics:}
These precoding objectives are distributed over the precoding layers such that $\bF_{c}^{(1)}$ requires slower time-varying channel state information compared with $\bF_{c}^{(2)}$, which in turn requires slower channel state information compared with $\bF_{c}^{(3)}$. Given the successive dimensionality reduction criteria, this means that the first precoding layer, which needs to be designed based on the large channel matrix, requires very large-scale channel statistics and needs to be updated every very long period of time. Similarly, the second and third precoding layers, which are designed based on the effective channels that have less dimensions, need to be updated more frequently.
\end{itemize}

In the next sections, we will present a specific multi-layer precoding design for full-dimensional massive MIMO systems, and show how it enables leveraging the large-scale MIMO channel characteristics to manage different kinds of interference with limited channel knowledge. We will also show how the multiplicative and successive reduced dimension structure of multi-layer precoding allows for efficient implementations using hybrid analog/digital architectures.

\section{Proposed Multi-Layer Precoding Design} \label{sec:Algorithm}

In this section, we present a multi-layer precoding algorithm for the full-dimensional massive MIMO system and channel models described in \sref{sec:Model}. Following the proposed multi-layer precoding criteria explained in \sref{sec:Concept}, we propose to design the $N_\mathrm{V} N_\mathrm{H} \times K$ precoding matrix $\bF_b$ of cell $b$, $b=1,...,B$ as 
\begin{align} \label{eq:prec_layer}
\bF_b&=\bF_b^{(1)} \bF_b^{(2)} \bF_b^{(3)}, 
\end{align}
where the first precoding layer $\bF^{(1)}_b$ is dedicated to avoid the out-of-cell interference, the second precoding layer $\bF_b^{(2)}$ is designed to maximize the effective signal power, and the third layer $\bF_b^{(3)}$ is responsible for canceling the intra-cell multi-user interference. Writing the received signal at user $k$ in cell $c$ in terms of the multi-layer precoding in \eqref{eq:prec_layer}, we get
\begin{equation} \label{eq:Received_0}
y_{c k}=\underbrace{\bh_{c c k}^* \bF^{(1)}_{c} \bF^{(2)}_{c} \bF^{(3)}_{c}\bs_{c}}_\text{received signal from serving BS}+ \underbrace{\sum_{b \neq c} \bh_{b c k}^* \bF^{(1)}_{b} \bF^{(2)}_{b} \bF^{(3)}_{b} \bs_{b}}_\text{received signal from other BSs} + n_{c k}.
\end{equation}  

Next, we explain in detail the proposed design of each precoding layer as well as the required channel knowledge.

\subsection{First Layer: Inter-Cell Interference Management}
We will design the first precoding layer ${\bF_b}^{(1)}$ to avoid the inter-cell interference, i.e., to cancel the second term of \eqref{eq:Received_0}. Exploiting the Kronecker structure of the channel model in \eqref{eq:channel}, we propose to construct the first layer as
\begin{equation} \label{eq:First_layer}
\bF^{(1)}_b={\bF^\mathrm{A}_b}^{(1)} \otimes {\bF^\mathrm{E}_b}^{(1)}.
\end{equation}

\noindent Adopting the channel model in \eqref{eq:channel} with $\overline{\bw}_{b c k}=\left({\boldsymbol{\Lambda}_{b c k}^{\mathrm{A}}}^{\frac{1}{2}} \otimes {\boldsymbol{\Lambda}_{b c k}^{\mathrm{E}}}^{\frac{1}{2}}\right) \bw_{b c k}$ and employing the Kronecker precoding structure in \eqref{eq:First_layer}, the second term of the received signal $y_{c k}$ in \eqref{eq:Received_0} can be expanded as
\begin{equation}
\sum_{b \neq c} \bh_{b c k}^* \bF^{(1)}_{b} \bF^{(2)}_{b} \bF^{(3)}_{b} \bs_{b}= \sum_{b \neq c}  \overline{\bw}_{b c k}^* \left(\bU_{b c k}^{\mathrm{A}^*} {\bF_b^\mathrm{A}}^{(1)} \otimes \bU_{b c k}^{\mathrm{E}^*} {\bF_b^\mathrm{E}}^{(1)} \right)  \bF^{(2)}_{b} \bF^{(3)}_{b} \bs_{b}.
\end{equation}

\noindent Avoiding the inter-cell interference for the users at cell $c$ can then be satisfied if $\bF_b^{(1)}, b \neq c$ is designed such that $ \bU_{b c k}^{\mathrm{E}^*} {\bF_b^\mathrm{E}}^{(1)}=\boldsymbol{0}, \forall k$. 
Equivalently, for any cell $c$ to avoid making interference on the other cell users, it designs its precoder ${\bF_c^\mathrm{E}}^{(1)}$ to be in the null-space of the elevation covariance matrices of all the channels connecting BS $c$ and the other cell users, i.e., to be in $\mathrm{Null}\left(\sum_{b \neq c} \sum_{k\in \cK_b} \bR_{c b k}^\mathrm{E}\right)$ with $\cK_b$ denoting the subset of $K$ scheduled users in cell $b$. Note that with large numbers of vertical antennas and for several channel models, this elevation inter-cell interference covariance will not have full rank and may actually have just a small overlap with the desired users' channels, as will be shown in \sref{sec:Perf}. 

Thanks to the directional structure of large-scale MIMO channels, we further note that with a large number of vertical antennas, $N_\mathrm{V}$, the null-space $\mathrm{Null}\left(\sum_{b \neq c} \sum_{k\in \cK_b} \bR_{c b k}^\mathrm{E}\right)$ of different scheduled users $\cK_b$ will have a large overlap. This means that designing ${\bF_c^\mathrm{E}}^{(1)}$ based on the  interference covariance subspace averaged over different scheduled users may be sufficient. Leveraging this intuition relaxes the required channel knowledge to design the first precoding layer. Hence, we define the average interference covariance matrix for BS $c$ as
\begin{equation}
\bR_{c}^{\mathrm{I}}=\sum_{b \neq c} \bbE_{\cK_b} \left[\bR^{E}_{c b k}\right].
\label{eq:Interference}
\end{equation}

\noindent In this paper, we manage the inter-cell interference in the elevation space, and therefore, set ${\bF^\mathrm{A}_b}^{(1)}=\bI_{N_\mathrm{H}}$. Let $\left[\bU_{c}^\mathrm{I} \ \bU_{c}^\mathrm{NI}\right] \boldsymbol{\Lambda}_c \left[\bU_{c}^\mathrm{I} \bU_{c}^\mathrm{NI}\right]^*$ represent the eigen-decomposition of $\bR_{c}^{\mathrm{I}}$ with the $N_\mathrm{v}\times r_\mathrm{I}$ matrix $\bU_{c}^\mathrm{I}$ and $N_\mathrm{v}\times r_\mathrm{NI}$ matrix $\bU_{c}^\mathrm{NI}$ corresponding to the non-zero and zero eigenvalues, respectively. Then, we design the first precoding layer $\bF_c^{(1)}$ to be in the null-space of the average interference covariance matrix by setting
\begin{equation} \label{eq:first_layer2}
\bF_c^{(1)}=\bI_{N_\mathrm{H}} \otimes \bU_{c}^\mathrm{NI},
\end{equation}
which is an $N_\mathrm{V} N_\mathrm{H} \times \rN N_\mathrm{H}$ matrix.

Given the design of the first precoding layer in \eqref{eq:first_layer2}, and defining the $\rN \times \rE$ effective elevation eigen matrix $\overline{\bU}_{c c k}^\mathrm{E}={\bUN_{c}}^* \bU^\mathrm{E}_{c c k}$, the received signal at user $k$ of cell $c$ in \eqref{eq:Received_0} becomes
\begin{equation}
y_{c k} = \overline{\bw}_{c c k}^*  \left(\bU_{c c k}^{\mathrm{A}^*} \otimes \overline{\bU}_{c c k}^{\mathrm{E}^*}\right)  \bF^{(2)}_{c} \bF^{(3)}_{c} \bs_{c}+ n_{c k}.
\label{eq:Received_2}
\end{equation}
\noindent Note that the first precoding layer in \eqref{eq:First_layer} acts as a spatial filter that entirely eliminates the inter-cell interference in the elevation domain. This filter, however, may have a negative impact on the desired signal power for the served users at cell $c$ if they share the same elevation subspace with the out-of-cell users. Therefore, this first layer precoding design is particularly useful for systems with low-rank elevation subspaces. It is worth mentioning here that recent measurements of 3D channels show that elevation eigenspaces may have low ranks at both low-frequency and millimeter wave systems \cite{3GPP_LTE,Akdeniz2014,Hur2016}. Relaxations of the precoding design in \eqref{eq:First_layer} are proposed in \sref{sec:Discuss} to compromise between inter-cell interference avoidance and desired signal power degradation.

\textbf{Required channel knowledge:} The design of the first precoding layer in \eqref{eq:First_layer} requires only the knowledge of the interference covariance matrix \textit{averaged} over different scheduled users. It depends therefore on very large time-scale channel statistics, which means that this precoding layer needs to be updated every very long period of time. This makes its acquisition overhead relatively small from an overall system perspective. In fact, this is a key advantage of the decoupled multi-layer precoding structure that allows dedicating one layer for canceling the out-of-cell interference based on large time-scale channel statistics while leaving the other layers to do other functions based on different time scales. This can not be done by typical precoding approaches that rely on one precoding matrix to manage different precoding objectives, as this precoding matrix will likely need to be updated based on the fastest channel statistics.

\subsection{Second Layer: Desired Signal Beamforming}
The second precoding layer $\bF_c^{(2)}$ is designed to focus the transmitted power on the served users' effective subspaces, i.e., on the user channels' subspaces including the effect of the first precoding layer. If we define the matrix consisting of the effective eigenvectors of user $k$ in cell $c$ as $\overline{\bU}_{c c k}=\left(\bU_{c c k}^{\mathrm{A}} \otimes \overline{\bU}_{c c k}^{\mathrm{E}}\right)$, then we design the second precoding layer $\bF_c^{(2)}$ as a large-scale conjugate beamforming matrix, i.e., we set
\begin{equation} \label{eq:second_layer}
\bF_c^{(2)}=\left[\overline{\bU}_{c c 1}, ..., \overline{\bU}_{c c K}\right],
\end{equation}
which has $N_\mathrm{H} \rN \times K \rA \rE$ dimensions. Given the second precoding layer design, and defining $\bG_{c,(k,r)}=\overline{\bU}_{c c k}^* \overline{\bU}_{c c r}$, the received signal by user $k$ in cell $c$ can be written as
\begin{equation}
y_{c k} = \overline{\bw}_{c c k}^* \left[\bG_{c,(k,1)}, ..., \bG_{c,(k,K)}\right] \bF_{c}^{(3)} \bs_c + n_{c k}.
\label{eq:Received_2x}
\end{equation}
The main objectives of this precoding layer can be summarized as follows. First, the effective channel vectors, including the first and second precoding layers, will have reduced dimensions compared to the original channels, especially when large numbers of antennas are employed. This reduces the overhead associated with training the effective channels, which is particularly important for FDD systems \cite{Adhikary2013,Alkhateeb2014b}. Second, this precoding layer supports the multiplicative structure  of multi-layer precoding with successive dimensionality reduction, which simplifies its implementation using hybrid analog/digital architectures, as will be briefly discussed in \sref{sec:Discuss}. 

\textbf{Required channel knowledge:} The design of the second precoding layer requires only the knowledge of the effective eigenvector matrices $\overline{\bU}_{c c k}, k=1, ..., K$, which depends on the large-scale channel statistics. It is worth noting that during the uplink training of the matrices $\overline{\bU}_{c c k}$, the first precoding layer works as spatial filtering for the other cell interference. Hence, this reduces (and ideally eliminates) the channel estimation error due to pilot reuse among  cells, and consequently leads to a pilot decontamination effect. 

\subsection{Third Layer: Multi-User Interference Management}
The third precoding layer $\bF_c^{(3)}$ is designed to manage the multi-user interference based on the effective channels, i.e., including the effect of the first and second precoding layers. If we define the effective channel of user $k$ in cell $c$ as $\overline{\bh}_{c k}=\left[\bG_{c,(k,1)}, ..., \bG_{c,(k,K)}\right]^* \overline{\bw}_{c c k}$, and let $\overline{\bH}_c=[\overline{\bh}_{c 1}, ..., \overline{\bh}_{c K}]$, then we construct the third precoding layer $\bF_c^{(3)}$ as a zero-forcing matrix
\begin{equation} \label{eq:third_layer}
\bF_c^{(3)}=\overline{\bH}_c\left(\overline{\bH}_c^* \overline{\bH}_c\right)^{-1} \bUpsilon_{c},
\end{equation}
where $\bUpsilon_{c}$ is a diagonal power normalization matrix that ensures satisfying the precoding power constraint $\|\left[\bF_b\right]_{:,k}\|^2=1$. Note that this zero-forcing design requires $N_\mathrm{H} \rN \geq K \rA \rE$, which is satisfied with high probability in massive MIMO systems, especially with sparse and low-rank channels. Given the design of the precoding matrix $\bF_c^{(3)}$, the received signal at user $k$ in cell $c$ can be expressed as
\begin{equation}
y_{c k} =  \left[\bUpsilon_{c}\right]_{k,k} s_{c,k} + n_{c k}.
\label{eq:Received_3}
\end{equation}

\textbf{Required channel knowledge:}
The design of the third precoding layer relies on the instantaneous effective channel knowledge. Thanks to the first and second precoding layers, these effective channels should have much smaller dimensions compared to the original channels in massive MIMO systems, which reduces the required training overhead.

\section{Performance Analysis} \label{sec:Perf}
The proposed multi-layer precoding design in \sref{sec:Algorithm} eliminates inter-cell interference as well as multi-user intra-cell interference, assuming that every BS $b$ has the knolwedge of its users' effective channels and channel covariance $\overline{\bH}_b$, $\left\{\bR_{cck}\right\}$ and the averaged inter-cell interference covariance in the elevation dimension $\bR_c^\mathrm{I}$. This interference cancellation, however, may have a penalty on the desired signal power which is implicitly captured by the power normalization factor $\left[\bUpsilon_c\right]_{k,k}$ in \eqref{eq:Received_3}. In this section, we will first characterize the achievable rate by the proposed multi-layer precoding design for a general channel model in Lemma \ref{lem:Ach_Rate}. Then, we will show that this precoding design can achieve optimal performance for some special yet important channel models in Section \ref{subsec:OneRing} and \sref{subsec:Physical}.  

\begin{lemma} \label{lem:Ach_Rate}
Consider the system and channel models in \sref{sec:Model} and the multi-layer precoding design in \sref{sec:Algorithm}. The achievable rate by user $k$ in cell $c$ is given by
\begin{equation} \label{eq:Ach_Rate_1}
R_{ck}=\log_2\left(1+ \frac{\mathsf{SNR}}{\left( \bW^*_c {\bF_{c}^{(2)}}^* {\bF_{c}^{(2)}} {\bF_{c}^{(2)}}^* {\bF_{c}^{(2)}}  \bW_c \right)_{k,k}^{-1}}\right),
\end{equation}	
where $\bW_c=\bI_{K} \circ \left[\overline{\bw}_{c c 1}, ..., \overline{\bw}_{c c K}\right]$ and $\mathsf{SNR}=\frac{P}{K \sigma^2}$.
\end{lemma}
\begin{proof}
See Appendix \ref{app:Ach_Rate}
\end{proof}

\noindent Note that the achievable rate in \eqref{eq:Ach_Rate_1} is upper bounded by the single-user rate---the rate when the user is solely served in the network---which is given by $\overline{R}_{ck}=\log_2\left(1+\mathsf{SNR} \left\|\overline{\bw}_{c c k}\right\|^2\right)$.  Therefore, Lemma \ref{lem:Ach_Rate} indicates that the proposed multi-layer precoding can achieve optimal performance if ${\bF_c^{(2)}}^*{\bF_c^{(2)}}=\bI$. To achieve that, it is sufficient to satisfy the following two conditions.
\begin{enumerate}[(i)] 
	\item {$\bG_{c,(k,m)}=\boldsymbol{0}, \forall m\neq k$, a condition that captures the impact of multi-user interference cancellation on the desired signal power.}
	\item {$\bG_{c,(k,k)}=\left({\bU_{c c k}^{\mathrm{A}}}^* \otimes {{\bU}_{c c k}^{\mathrm{E}}}^*\right){\bF_c^{(1)}} {\bF_c^{(1)}}^*\left({\bU_{c c k}^{\mathrm{A}}} \otimes {\bU}_{c c k}^{\mathrm{E}}\right)=\bI, \forall k$, a condition that captures the possible impact of the inter-cell interference avoidance on the desired signal power.}
\end{enumerate} 

Next, we characterize the performance of multi-layer precoding for two special yet important channel models, namely, the one-ring and single-path channel models.

\subsection{Performance with One-Ring Channel Models} \label{subsec:OneRing}

Motivated by its analytical tractability and meaningful geometrical interpretation, we will consider the one-ring channel model in this subsection \cite{Shiu2000,Petrus2002,Abdi2002,Zhang2007}. This will enable us to draw useful insights into the performance of multi-layer precoding, which can then be extended to more general channel models. Note that due to its tractability, one-ring channel models have also been adopted in prior massive MIMO work \cite{Bjoernson2014,Adhikary2013,Yin2013,Shen2016}.

\begin{figure}[h]	
	\centering
	\includegraphics[width=.45\columnwidth]{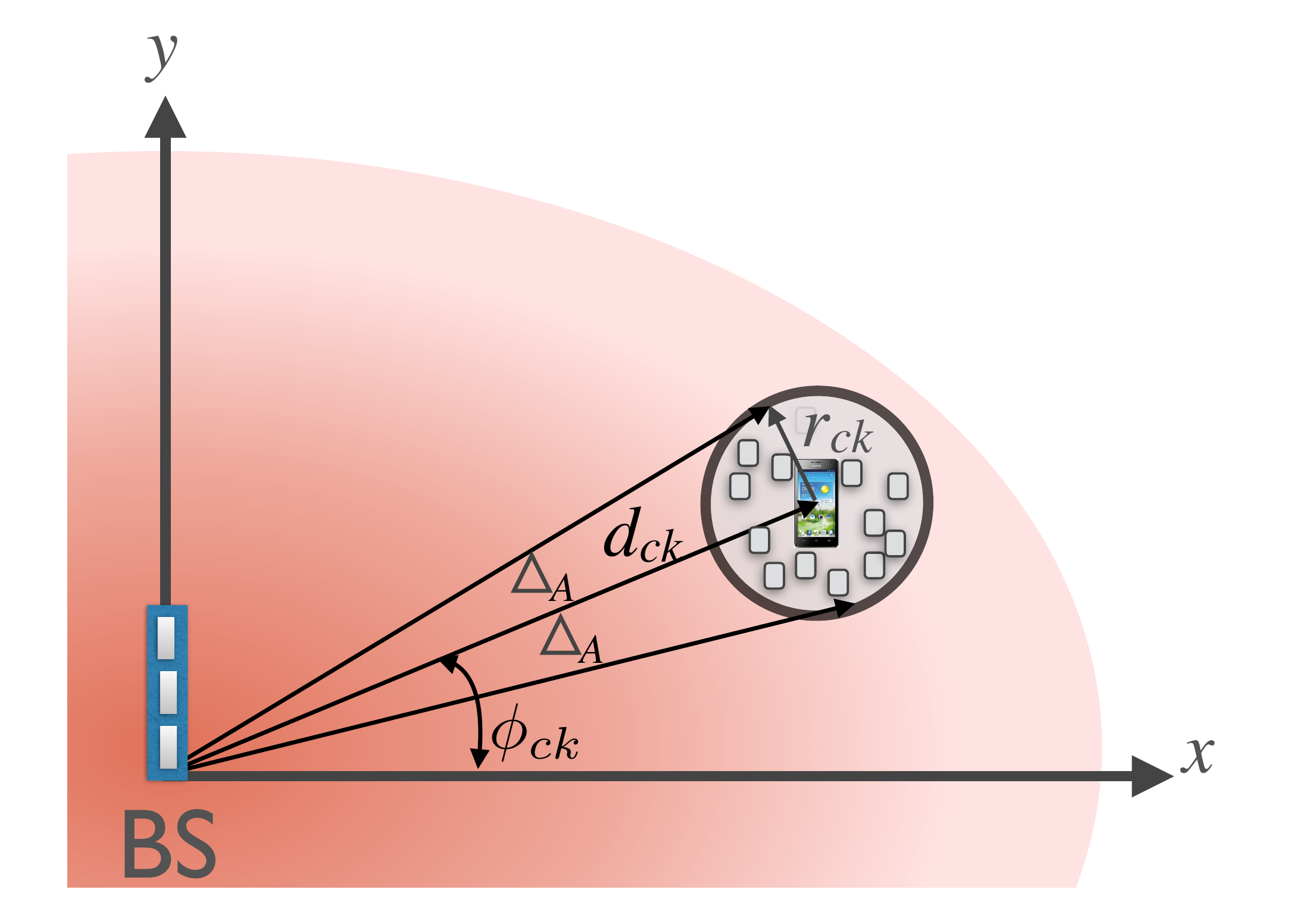}
	\caption{An illustration of the one-ring channel model in the azimuth direction. The BS that, has a UPA in the y-z plane, serves a mobile user in the x-y plane at distance $d_{c k}$. The user is surrounded by scatterers on a ring of radius $r_{c k}$, and its channel experiences an azimuth angular spread $\Delta_\rm{A}$.}
	\label{fig:OneRing}
\end{figure}

The one-ring channel model describes the case when a BS is elevated away from scatterers and is communicating with a mobile user that is surrounded by a ring of scatterers. Consider a BS at height $H_\mathrm{BS}$ employing an $N_V \times N_H$ UPA, and serving a mobile user at a distance $d_{ck}$  with azimuth and elevation angles $\phi_{ck}, \theta_{ck}$, as depicted in \figref{fig:OneRing}. If the mobile is surrounded by scatterers on a ring of radius $r_{ck}$ in the azimuth dimension, then the azimuth angular spread $\Delta_\mathrm{A}$ can be approximated as $\Delta_\mathrm{A}=\arctan{\left(\frac{r_{ck}}{d_{ck}}\right)}$. Further, assuming for simplicity that the received power is uniformly distributed over the ring, then the correlation between any two antenna elements with orders $n_1, n_2$ in the horizontal direction is given by 
\begin{equation} \label{eq:A_Cov}
\left[\bR^{A}_{cck}\right]_{n_1,n_2}=\frac{1}{2 \Delta_\mathrm{A}}\int_{-\Delta_\mathrm{A}}^{\Delta_\mathrm{A}} {e^{- j \frac{2 \pi}{\lambda} d (n_2-n_1) \sin(\phi_{ck}+\alpha)\sin(\theta_{ck})} d\alpha}.
\end{equation}

\noindent The elevation correlation matrix can be similarly defined for the user $k$, in terms of its elevation angular spread $\Delta_\mathrm{E}$.

In the next theorem, we characterize the achievable rate for an arbitrary user $k$ in cell $c$ under the one-ring channel model.

\begin{theorem} \label{th:CellCenter}
Consider the full-dimensional cellular system model in \sref{sec:SysModel} with cells of radius $r_\mathrm{cell}$, and the channel model in \sref{sec:ChModel} with the one-ring correlation matrices in \eqref{eq:A_Cov}. Let $\phi_{c k}, \theta_{c k}$ denote the azimuth and elevation angles of user $k$ at cell $c$, and let $\Delta_\mathrm{A}, \Delta_\mathrm{E}$ represent the azimuth and elevation angular spread. Define the maximum distance with no blockage on the desired signal power as $d_\mathrm{max}=H_\mathrm{BS} \tan\left(\arctan\left(\frac{r_\mathrm{cell}}{H_\mathrm{BS}}\right)- 2 \Delta_\mathrm{E}\right)$. If $\left|\phi_{c k}-\phi_{c m}\right| \geq 2 \Delta_\rm{A}$ or $\left|\theta_{c k}-\theta_{c m}\right| \geq 2 \Delta_\rm{E}$, $\forall m \neq k$, and $d_{c k} \leq d_\rm{max}$, then the achievable rate of user $k$ at cell $c$, when applying the multi-layer precoding algorithm in \sref{sec:Algorithm}, satisfies
\begin{equation} \label{eq:Rate_InCell}
\lim_{N_\rm{V}, N_\rm{H} \rightarrow \infty} R_{c k}=\overline{R}_{c k}=\log_2\left(1+\mathsf{SNR} \left\|\overline{\bw}_{c c k}\right\|^2 \right).
\end{equation} 
\end{theorem}
\begin{proof}
See Appendix \ref{app:CellCenter}
\end{proof}

Theorem \ref{th:CellCenter} indicates that the achievable rate with multi-layer precoding converges to the optimal single-user rate for the users that are not at the cell edge ($r_\rm{cell}-d_\mathrm{max}$ away from cell edge), provided that they maintain either an azimuth or elevation separation by double the angular spread. For example, consider a cellular system with cell radius $100$m and BS antenna height $50$m, if the elevation angular spread equals $\Delta_\rm{E}=3^\circ$, then all the users within $\sim 80$m distance from the BS achieve optimal rate. It is worth noting here that these rates do not experience any pilot contamination or multi-user interference impact and can, therefore, grow with the antenna numbers or transmit power without any bound on the maximum values that they can reach. 

The angular separation between the users in Theorem \ref{th:CellCenter} can be achieved via user scheduling techniques or other network optimization tools. In fact, even without user scheduling, this angular separation is achieved with high probability as will be illustrated by simulations in \sref{sec:Results} under reasonable system and channel assumptions.  Further, for sparse channels with finite number of paths, it can be shown that this angular separation is not required to achieve the optimal rate. Studying these topics are interesting future extensions.

In the following theorem, we derive a lower bound on the achievable rate with multi-layer precoding for the cell-edge users.

\begin{theorem}\label{th:CellEdge}
	Consider the system and channel models described in Theorem \ref{th:CellCenter}. If $\left|\phi_{c k}-\phi_{c m}\right| \geq 2 \Delta_\rm{A}$ or $\left|\theta_{c k}-\theta_{c m}\right| \geq 2 \Delta_\rm{E}$, $\forall m \neq k$, and $d_\rm{max} \leq d_{c k} \leq r_\rm{cell}$, then the achievable rate of user $k$ at cell $c$, when applying the multi-layer precoding algorithm in \sref{sec:Algorithm}, satisfies
	\begin{equation} \label{eq:Rate_CellEdge}
	\lim_{N_\rm{V}, N_\rm{H} \rightarrow \infty} R_{c k} \geq \log_2\left(1+\mathsf{SNR} \left\|\overline{\bw}_{c c k}\right\|^2 \sigma^2_\rm{min}\left(\overline{\bU}_{c c k}^\rm{E}\right)\right),
	\end{equation} 
	where $\sigma_\rm{min}\left(\bA\right)$ denotes the minimum singular value of the matrix $\bA$.
\end{theorem}
\begin{proof}
	See Appendix \ref{app:CellEdge}
\end{proof}

Theorem \ref{th:CellEdge} indicates that cell edge users experience some degradation in their $\mathsf{SNRs}$ as a cost for the perfect inter-cell interference avoidance. In \sref{sec:Discuss}, we will discuss some solutions that make compromises between the degradation of the desired signal power and the management of the inter-cell interference for cell-edge users, under the multi-layer precoding framework.

\subsection{Performance with Single-Path Channel Models} \label{subsec:Physical}
Rank-1 channel models describe the cases where the signal propagation through the channel is dominated by one line-of-sight (LOS) or non-LOS (NLOS) path. This is particularly relevant to systems with sparse channels, such as mmWave systems \cite{Bai2014,Rappaport2013a,Hur2016}. A special case of rank-1 channel models is the single-path channels. Consider a user $k$ at cell $c$ with a single path channel, defined by its azimuth and elevation angles $\phi_{c k}, \theta_{c k}$. Then, the channel vector can be expressed as
\begin{equation} \label{eq:Rank1_CH}
\bh_{c c k}=\rho^{\frac{1}{2}}_{cc  k} \ \beta_{c k} \ \ba_\rm{A}\left(\phi_{c k}, \theta_{c k}\right) \otimes \ba_\rm{E}\left(\phi_{c k}, \theta_{c k}\right), 
\end{equation} 
where $\ba_\rm{A}\left(\phi_{c k}, \theta_{c k}\right)$ and $\ba_\rm{E}\left(\phi_{c k}, \theta_{c k}\right)$ are the azimuth and elevation array response vectors, $\beta_{c k}$ is the complex path gain, and $\rho_{c c k}$ is its path loss.

In the next corollary, we characterize the achievable rate of the proposed multi-layer precoding design for single-path channels. 

\begin{corollary} \label{cor:Rank1}
	Consider the full-dimensional cellular system model in \sref{sec:SysModel}, and the single-path channel model in \eqref{eq:Rank1_CH}. When applying the multi-layer precoding algorithm in \sref{sec:Algorithm}, the achievable rate of user $k$ at cell $c$ satisfies
	\begin{equation}
	\lim_{N_\rm{V}, N_\rm{H} \rightarrow \infty} R_{c k}= \overline{R}_{c k} = \log_2\left(1+ \mathsf{SNR}  \left\|\bh_{ c c k}\right\|^2 \right).
	\end{equation}
\end{corollary}
\begin{proof}
	The proof is similar to that of Theorem \ref{th:CellCenter}, and is omitted due to space limitations.
\end{proof} 

Corollary \ref{cor:Rank1} indicates that the proposed multi-layer precoding design can achieve an optimal performance for single-path channels, making it a promising solution for mmWave and low channel rank massive MIMO systems. This will also be verified by numerical simulations in \sref{sec:Results}.

\section{Discussion and Extensions} \label{sec:Discuss}

While we proposed and analyzed a specific multi-layer precoding design in this paper, there are many possible extensions as well as important topics that need further investigations. In this section, we briefly discuss some of these points, leaving their extensive study for future work. 

\subsection{Multi-Layer Precoding with Augmented Vertical Dimensions} \label{subsec:Aug}

As explained in \sref{sec:Algorithm}, the proposed multi-layer precoding design attempts to perfectly avoid the inter-cell interference by forcing its transmission to be in the elevation null-space of the interference. While this guarantees optimal performance for cell-interior users and decontaminates the pilots for all the cell users, it may also block some of the desired signal power at the cell-edge. In this section, we propose a modified design for the first precoding layer $\bF_c^{(1)}$ that compromises between the inter-cell interference avoidance and the desired signal degradation. The main idea of the proposed design, that we call multi-layer precoding with augmented vertical dimensions, is to simply extend the null-space of the inter-cell interference via exploiting the structure of large channels. This is summarized as follows. Leveraging Lemma 2 in \cite{Yin2013}, the rank of the one-ring correlation matrix can be related to its angular range $\left[\theta_\rm{min}, \theta_\rm{max}\right]$ as 
\begin{equation}
\text{rank}\left(\bR\right)=\frac{N D}{\lambda} \left(\cos(\theta_\rm{min})-\cos(\theta_\rm{max})\right) \ \text{as} \ N \rightarrow \infty.
\end{equation}
\noindent Applying this lemma to the elevation inter-cell interference subspace, setting $\theta_\rm{min}=\pi/2$, BS $c$ can estimate its maximum interference elevation angle, denoted $\theta_c^\rm{I}$, as
\begin{equation}
\theta_c^\rm{I}=\arccos\left(-\frac{\text{rank}\left(\bR_c^\rm{I}\right) \lambda}{N_\rm{V} D}\right).
\end{equation}

\noindent Extending the null space of the interference can then be done by virtually reducing the inter-cell interference subspace. Let $\delta_\rm{E}$ denote the angular range of the extended subspace.  The modified inter-cell interference covariance can then be calculated as 
\begin{equation}
\left[\overline{\bR}_c^\rm{I}\right]_{n_1,n_2}=\frac{1}{\theta_c^\rm{I}-\delta_\rm{E}-\pi/2} \int_{\frac{\pi}{2}}^{\theta_c^\rm{I}-\delta_\rm{E}} e^{j k D (n_2-n_1) \cos(\alpha)} d\alpha.
\end{equation} 

\noindent Finally, if $\left[\overline{\bU}_{c}^\mathrm{I} \ \overline{\bU}_{c}^\mathrm{NI}\right] \overline{\boldsymbol{\Lambda}}_c \left[\overline{\bU}_{c}^\mathrm{I} \overline{\bU}_{c}^\mathrm{NI}\right]^*$ represents the eigen-decomposition of $\overline{\bR}_c^\rm{I}$, with $\overline{\bU}_{c}^\mathrm{I}$ and $\overline{\bU}_{c}^\mathrm{NI}$ correspond to the non-zero and zero eigenvalues, then the modified first precoding layer can be constructed as
\begin{equation}
\bF_{c}^{(1)}= \bI \otimes \overline{\bU}_c^\rm{NI}.
\end{equation}

Note that under this multi-layer precoding design, only cell edge users will experience inter-cell interference and pilot contamination while optimal performance is still guaranteed for cell-interior users. This yields an advantage for multi-layer precoding over conventional massive MIMO precoding schemes, which will also be illustrated by numerical simulations in \sref{sec:Results}.

\subsection{TDD and FDD Operation with Multi-Layer Precoding}

While we focused on TDD systems in this paper, the fact that multi-layer precoding relies on channel covariance knowledge makes it attractive for FDD operation as well. In FDD systems, the adjacent cells will cooperate to construct the elevation inter-cell interference subspace, which is needed to build the first precoding layer. Since this channel knowledge is of very large-scale statistics and this precoding layer needs to be updated every long time period, this cooperation overhead can be reasonably low. Given the first layer spatial filtering, every BS can estimate its users covariance knowledge free of inter-cell interference. Thanks to the multiplicative structure of the multi-layer precoding and its successive dimensionality reduction, only the third precoding layer requires the instantaneous knowledge of the effective channel, which has much smaller dimensions. It is worth noting here that other FDD massive MIMO precoding schemes, such as JSDM \cite{Adhikary2013} with its user grouping functions, can be easily integrated into the proposed multi-layer precoding framework for full-dimensional massive MIMO cellular systems.

In TDD systems, the required channel knowledge for the three stages can be done through uplink training on different time scales. One important note is that the second precoding layer (and its channel training) may not be needed in TDD systems with fully-digital transceivers, as the instantaneous channels can be easily trained in the uplink with a small number of pilots. This precoding layer, however, is important if multi-layer precoding is implemented using hybrid architectures, as will be shown in the following subsection.


\subsection{Multi-Layer Precoding using Hybrid Architectures} 

\begin{figure}[t]	
	\centering
	\includegraphics[width=.5\columnwidth]{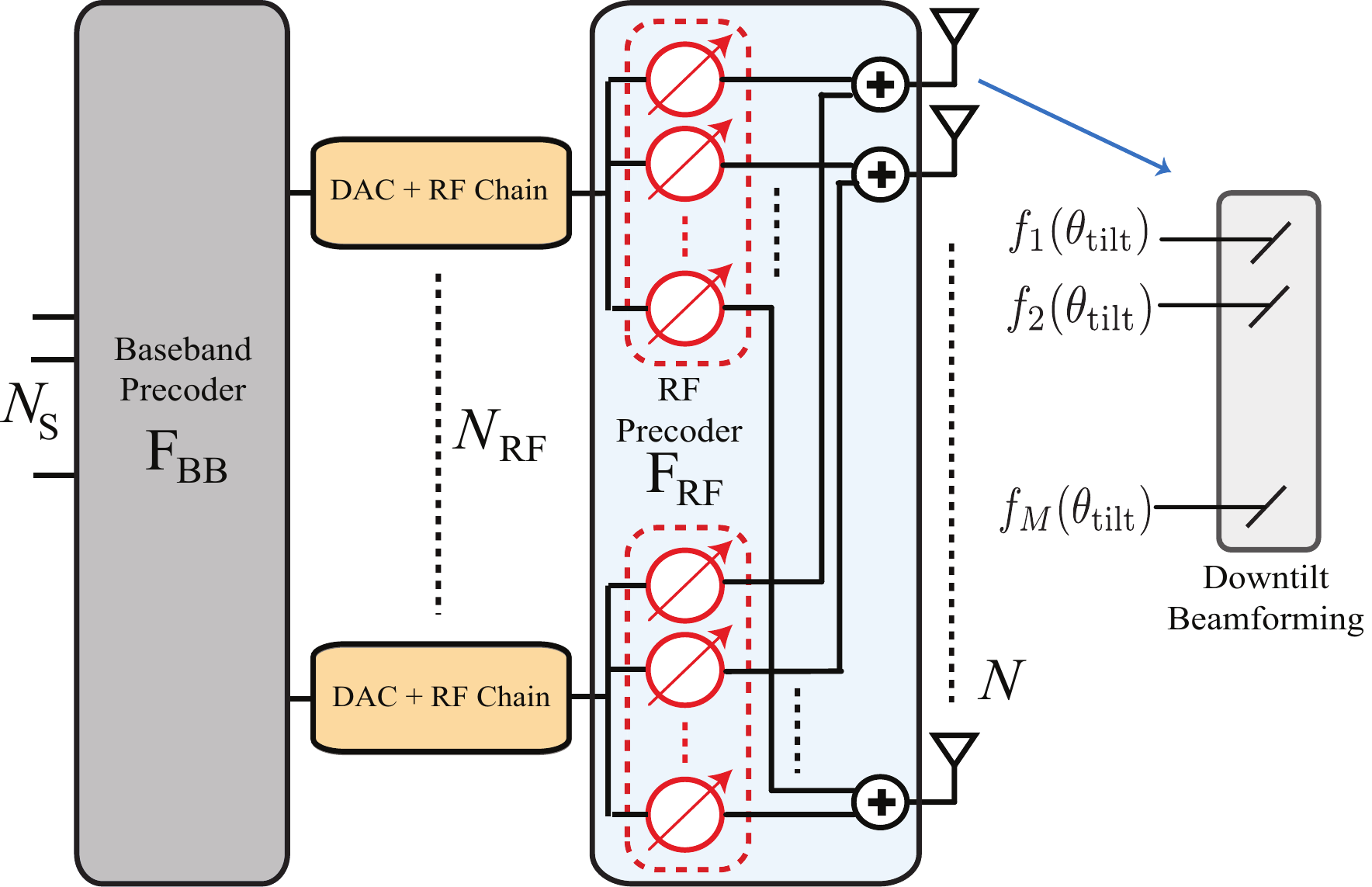}
	\caption{The figure shows a hybrid analog/digital architecture, at which baseband precoding, RF precoding, and antenna downtilt beamforming can be utilized to implement the multi-layer precoding algorithm.}
	\label{fig:Beam}
\end{figure}

Thanks to the multiplicative structure and the specific multi-layer precoding design in \sref{sec:Algorithm}, we note that each precoding layer has less dimensions compared to the prior layers. This allows the multi-layer precoding matrices to be implemented using hybrid analog/digital architectures\cite{HeathJr2016,ElAyach2014,Alkhateeb2014b,Han2015}, which reduces the required number of RF chains. In this section, we briefly highlight one possible idea for the hybrid analog/digital implementation, leaving its optimization and extensive investigation for future work.

Considering the three-stage multi-layer precoding design in \sref{sec:Algorithm}, we propose to implement the first and second layers in the RF domain and perform the third layer precoding at baseband, as depicted in \figref{fig:Beam}. Given the successive dimensional reductions, the required number of RF chains is expected to be much less than the number of antennas, especially in sparse and low-rank channels. As the first precoding layer focuses on avoiding the inter-cell interference in the elevation direction, we can implement it using downtilt directional antenna patterns. We assume that each antenna port has a directional pattern and electrically adjusted downtilt angle \cite{Kammoun2014, Seifi2014a}. For example, the 3GPP antenna port elevation gain $G^\mathrm{E}\left(\theta\right)$ is defined as \cite{Kammoun2014}
\begin{equation}
G^\mathrm{E}(\theta)=G^\mathrm{E}_\mathrm{max}-\min\left\{12 \left(\frac{\theta-\theta_\mathrm{tilt}}{\theta_\mathrm{3dB}}\right)^2, \mathrm{SL}\right\},
\end{equation}
where $\theta_\mathrm{tilt}$ is the downtilt angle, and SL is the sidelobe level. Therefore, one way to approximate ${\bF_c^{\mathrm{E}}}^{(1)}$ is to adjust the downtilt angle $\theta_\mathrm{tilt}$ to minimize the leakage transmission outside the interference null-space $\bU_c^\mathrm{NI}$.

Once $\bF_c^{(1)}$ is implemented, the second precoding layer $\bF_c^{(2)}$ can be designed similar to \cite{Alkhateeb2014b}, i.e., each column of  $\bF_c^{(2)}$ can be approximated by a beamsteering vector taken from a codebook that captures the analog hardware constraints. Finally, the third precoding layer $\bF_c^{(3)}$ is implemented in the baseband to manage the multi-user interference based on the effective channels.

\section{Simulation Results} \label{sec:Results}

In this section, we evaluate the performance of the proposed multi-layer precoding algorithm using numerical simulations. We also draw insights into the impact of the different system and channel parameters. 

We consider a single-tier 7-cell cellular system model as depicted in \figref{fig:Cells}, and calculate the performance for the cell in the center. Unless otherwise mentioned, every BS is assumed to a have a UPA, oriented in the y-z plane, at a height $H_\rm{BS}=35$m, and serving users at cell radius $r_\rm{cell}=100$m. Users are randomly and uniformly dropped in the cells, and every cell randomly schedules $K=20$ users to be served at the same time and frequency slot. The BS transmit power is $P=5$ dB and the receiver noise figure is $7$ dB. The system operates at a carrier frequency $4$ GHz, with a bandwidth $10$ MHz, and a path loss exponent $3.5$. Two channel models are assumed, namely, the single-path and the one-ring channel models. 

The BSs in the adopted system apply the multi-layer precoding algorithm in \sref{sec:Algorithm}. The required channel knowledge is perfectly obtained from the geometry of the network, i.e., no actual channel estimation is applied. We assume a universal pilot reuse, i.e., all the cells randomly assign the same $K$ orthogonal pilots to its users. The channels of the users sharing the same pilots are added at every BS, which simulates the interference of the other cells' users in the channel estimation phase. In more detail, the channel estimation and multi-layer precoding are done as follows. First, the averaged interference covariance matrix $\bR_b^\mathrm{I}$ of every BS $b$ is constructed by averaging the elevation interference covariance over $40$ realizations of scheduled users, each has $20$ users/cell. Using this knowledge, the first-layer precoders are obtained according to \eqref{eq:First_layer}. Then, the effective channel covariance matrices $\overline{\bU}_{c c k}$ are calculated by applying the first precoding layer on the sum of the channel covariance matrices of the users that share the same pilots from the different cells. Note that the first-layer precoders act as spatial filters that reduce (and ideally eleminate) the contributions of the other cells in the sum of the channel covariance. The second-layer precoders are then obtained following \eqref{eq:second_layer}. The effective channels are similarly calculated, applying the first and second precoding layers, from which the third-layer precoders are constructed using \eqref{eq:third_layer}. For the other precoding schemes we compare with, the channels are similarly constructed using the geometry and by adding the co-pilot user channels. Next, we present the simulation results for the two adopted channel models.
 
\begin{figure}[t]
	\centering
	\subfigure[center][{FD massive MIMO cellular model}]{
		\raisebox{.25\height}{\includegraphics[width=.4\columnwidth]{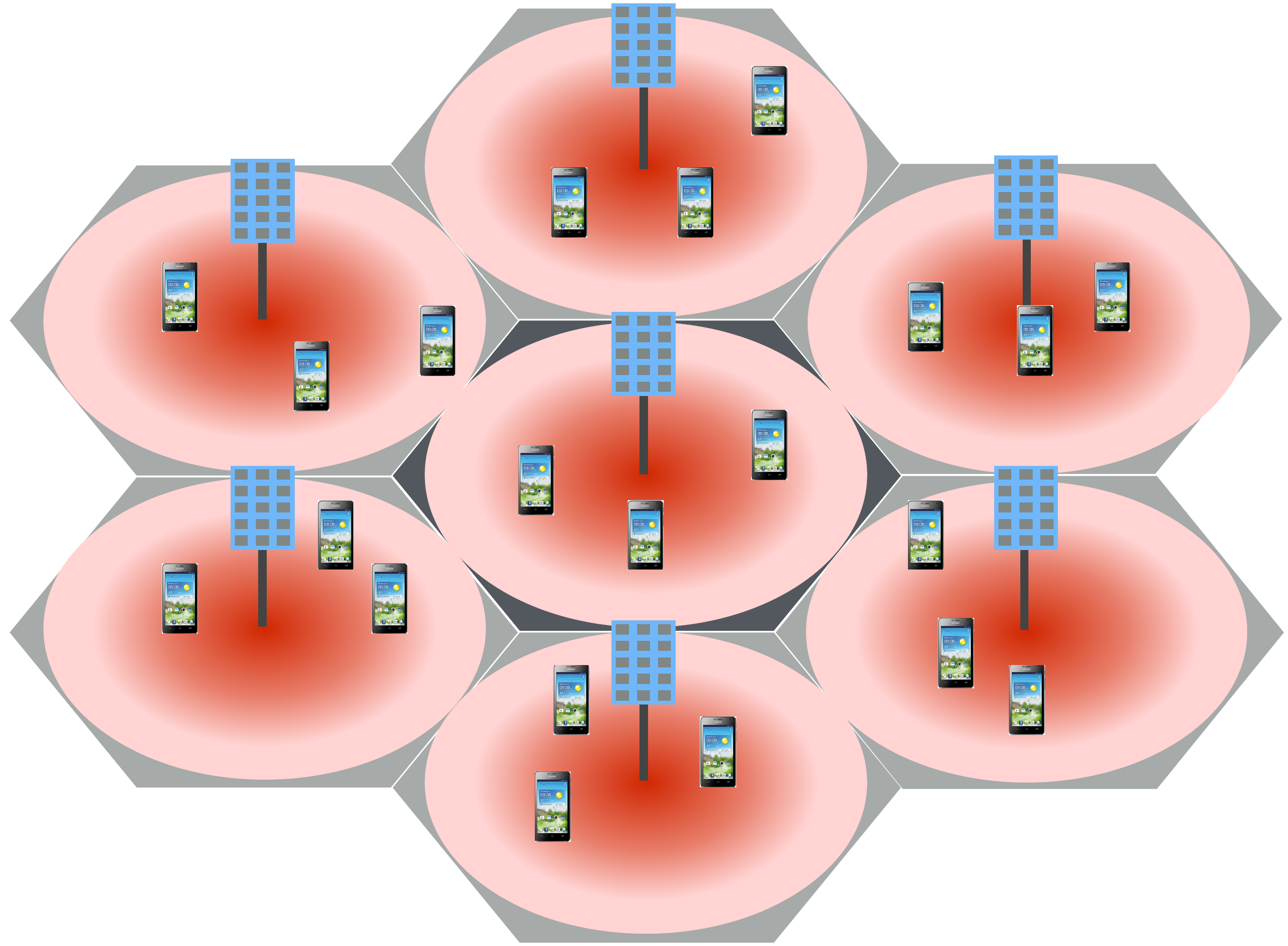}}
		\label{fig:Cells}}
	\subfigure[center][{$N_\rm{H}=30$}]{
		\includegraphics[width=.55\columnwidth]{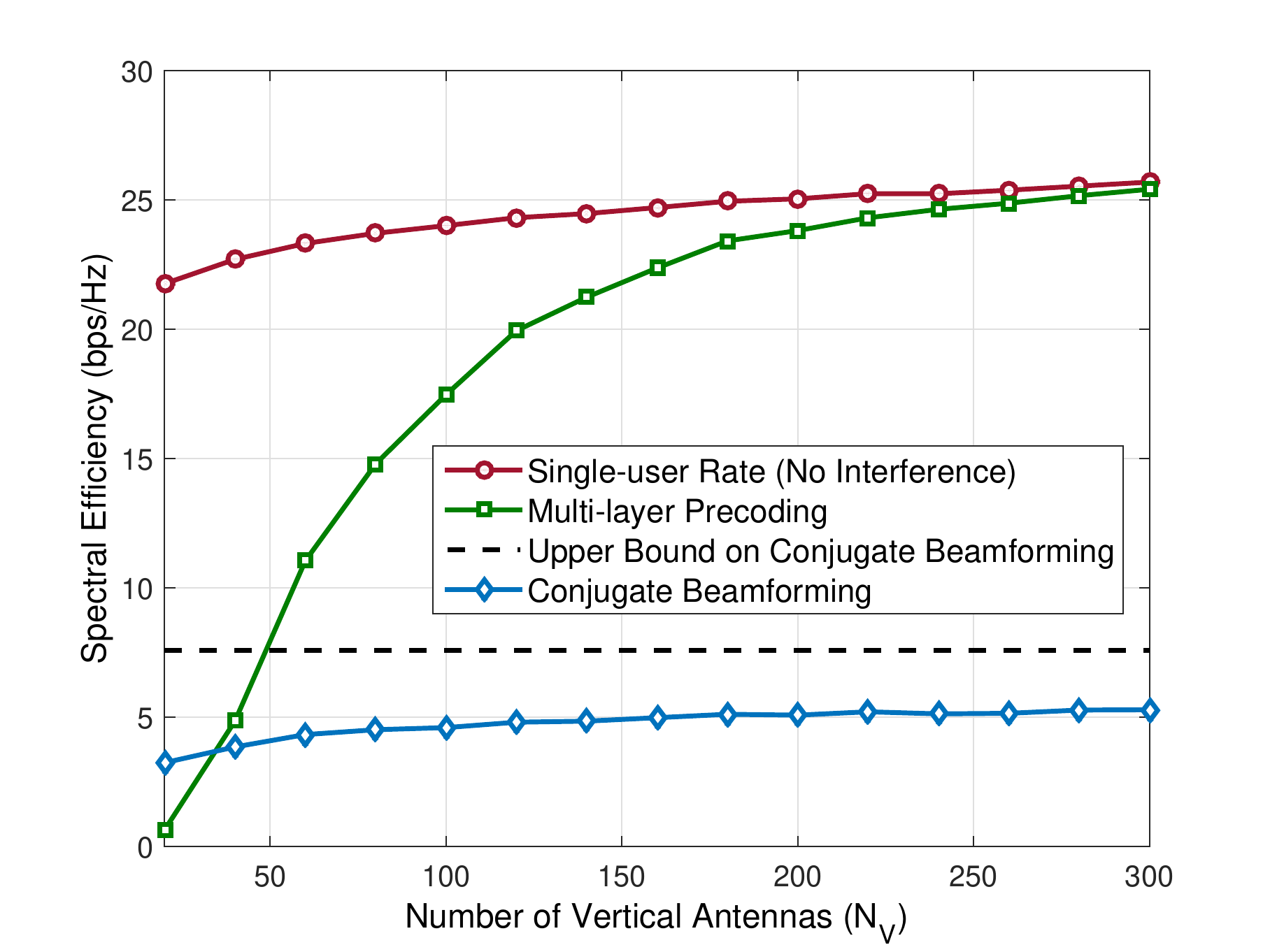}
		\label{fig:Ant_Rank1}}
	\caption{The adopted single-tier (7-cells) cellular model with FD massive MIMO antennas at the BSs is illustrated in (a). In (b), the achievable rate of the proposed multi-layer precoding is compared to the single-user rate and the rate with conventional conjugate beamforming, for different numbers of vertical antennas. The number of BS horizontal antennas is $N_\rm{H}=30$, and the users are assumed to have single-path channels.}
\end{figure}

\subsection{Results with Single-Path Channels}

In this section, we adopt a single-path model for the user channels as described in \eqref{eq:Rank1_CH}. The azimuth and elevation angles are geometrically determined based on users' locations relative to the BSs, and the complex path gains $\beta_{c k} \sim \mathcal{CN} \left(0,1\right)$. 

\begin{figure}[t]
	\centering
	\subfigure[center][{$r_\rm{cell}=200$m}]{
		\includegraphics[width=.482\columnwidth,height=175pt]{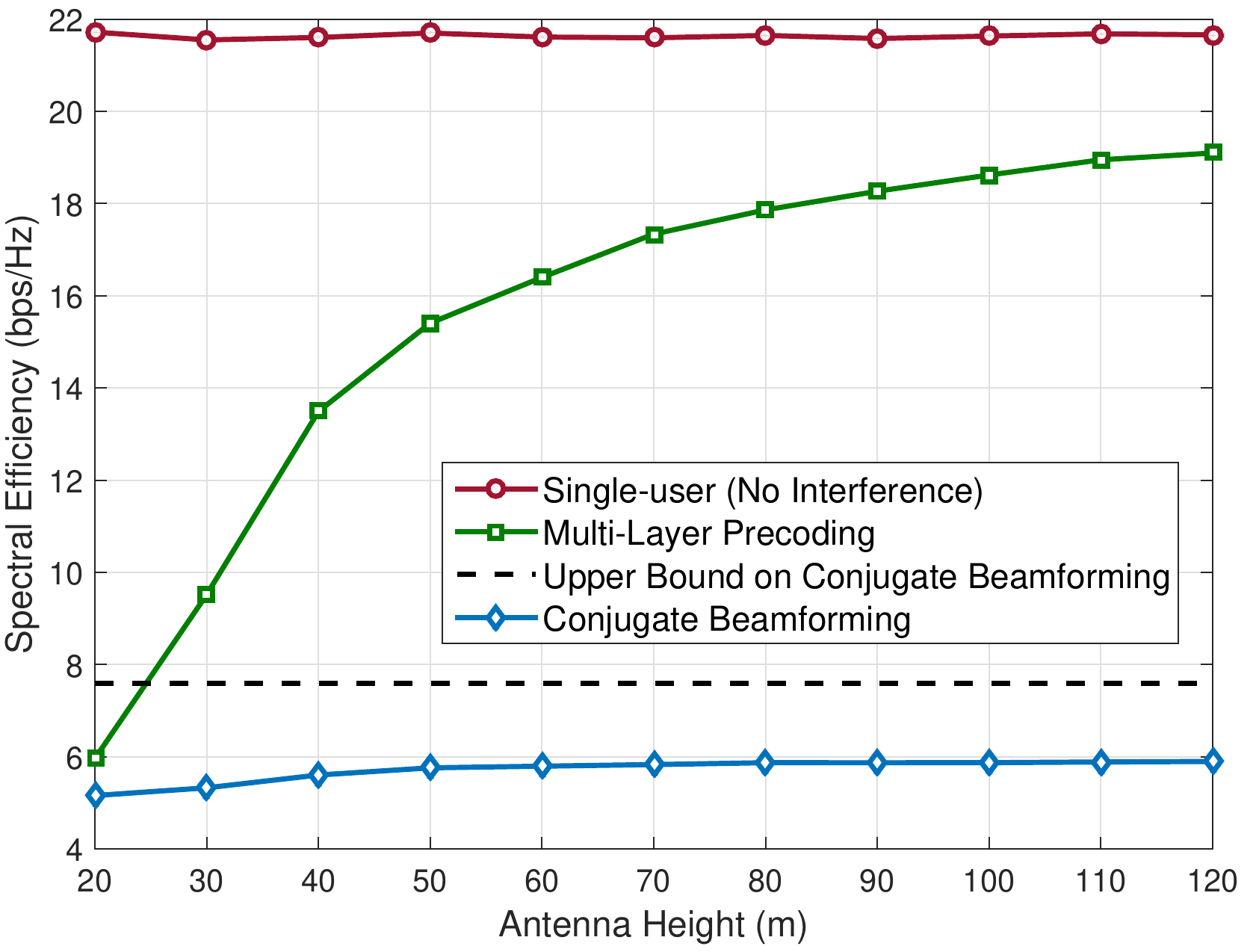}
		\label{fig:Imact_Hei}}
	\subfigure[center][{$H_\rm{BS}=35$m}]{
		\includegraphics[width=.482\columnwidth, height=175pt]{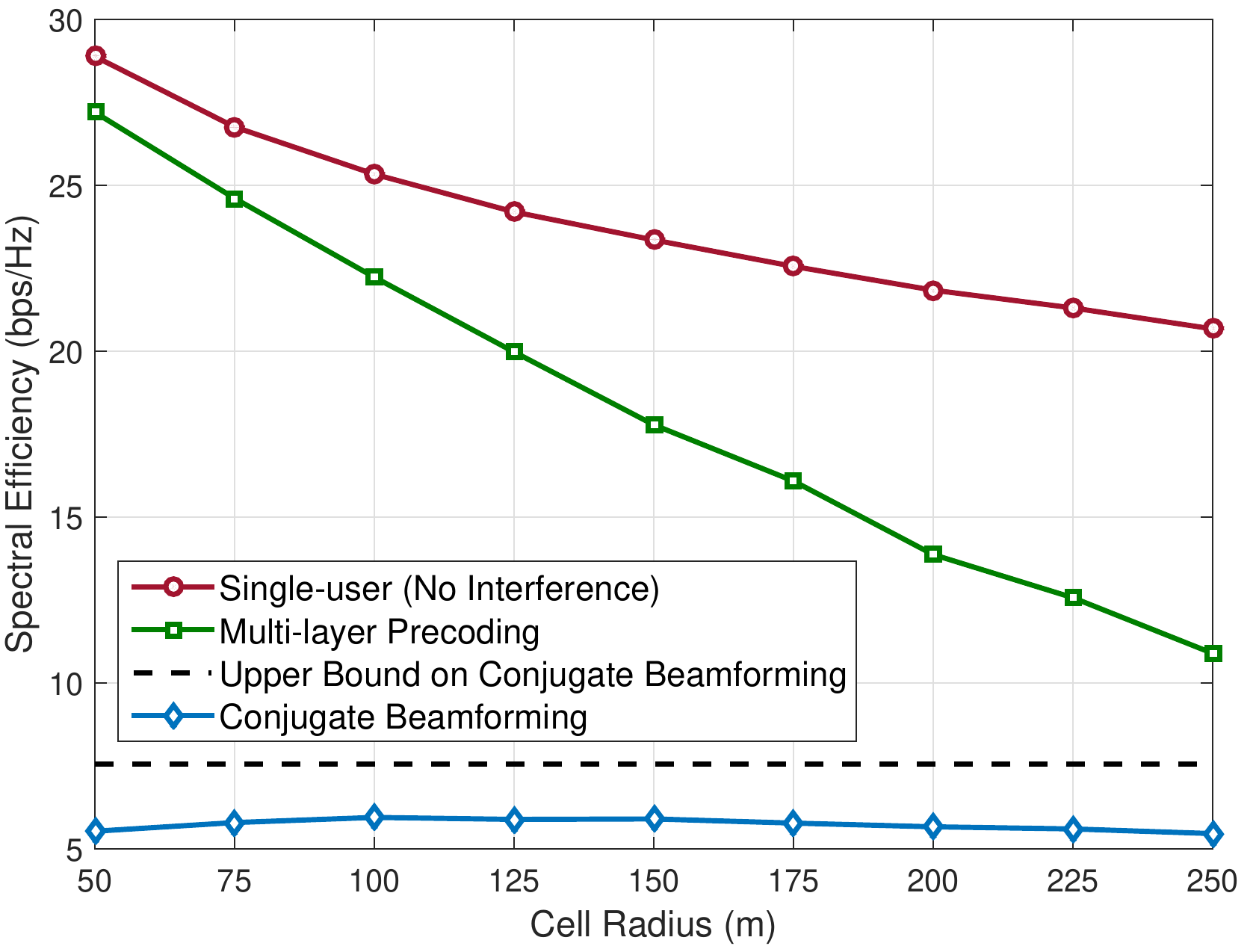}}
	\label{fig:Imact_Red}
	\caption{The achievable rate of the proposed multi-layer precoding is compared to the single-user rate and the rate with conventional conjugate beamforming for different cell radii. The BSs are assumed to employ $120 \times 30$ UPAs, and the users have single-path channels.}
	\label{fig:Imact_Rad_Hei}
\end{figure}

\textbf{Optimality with large antennas:} 
In \figref{fig:Ant_Rank1}, we compare the per-user achievable rate of multi-layer precoding with the single-user rate and the rate with conventional conjugate beamforming. The BSs are assumed to employ UPAs that have $N_\rm{H}=30$ horizontal antennas and different numbers of vertical antennas. First, we note that the per-user achievable rate with multi-layer precoding approaches the optimal single-user rate as the number of antennas grow large. This verifies the asymptotic optimality result of multi-layer precoding given in Corollary \ref{cor:Rank1}. Note that the single-user rate is the rate if only this user is served in the network, i.e., with no inter-cell or multi-user intra-cell interference. In the figure, we also plot the achievable rate with conventional conjugate beamforming. This assumes that channels are estimated using uplink training and then conjugate beamforming is applied in the downlink data transmission \cite{Marzetta2010}. As a function of the path-loss $\rho_{b c k}$ in \eqref{eq:Rank1_CH}, the conjugate beamforming rate is theoretically bounded from above by \cite{Marzetta2010}
\begin{equation}
\overline{R}_{c k}^{CB}= \log_2\left(1+\mathsf{SNR} \frac{\rho_{c c k}^2}{\sum_{b \neq c} \rho_{b c k}^2}\right),
\end{equation} 
which limits its rate from growing with the number of antennas beyond this value. Interestingly, the multi-layer precoding rate does not have a limit on its rate and can grow with the number of antennas and transmit power without a theoretical limit. The intuition behind that lies in the inter-cell interference avoidance using averaged channel covariance knowledge in multi-layer precoding. This works as a spatial filtering that avoids uplink channel estimation errors due to pilot reuse among cells and cancels inter-cell interference in the downlink data transmission. Therefore, the multi-layer precoding rate is free of the pilot-contamination impact. Note that while the asymptotic optimality of multi-layer precoding is realized at large antenna numbers, \figref{fig:Ant_Rank1} shows it can still achieve gain over conventional massive MIMO beamforming schemes at much lower number of antennas. 

\textbf{Impact of antenna heights and cell radii:}
\begin{figure}[t]	
	\centering
	\includegraphics[width=.7\columnwidth]{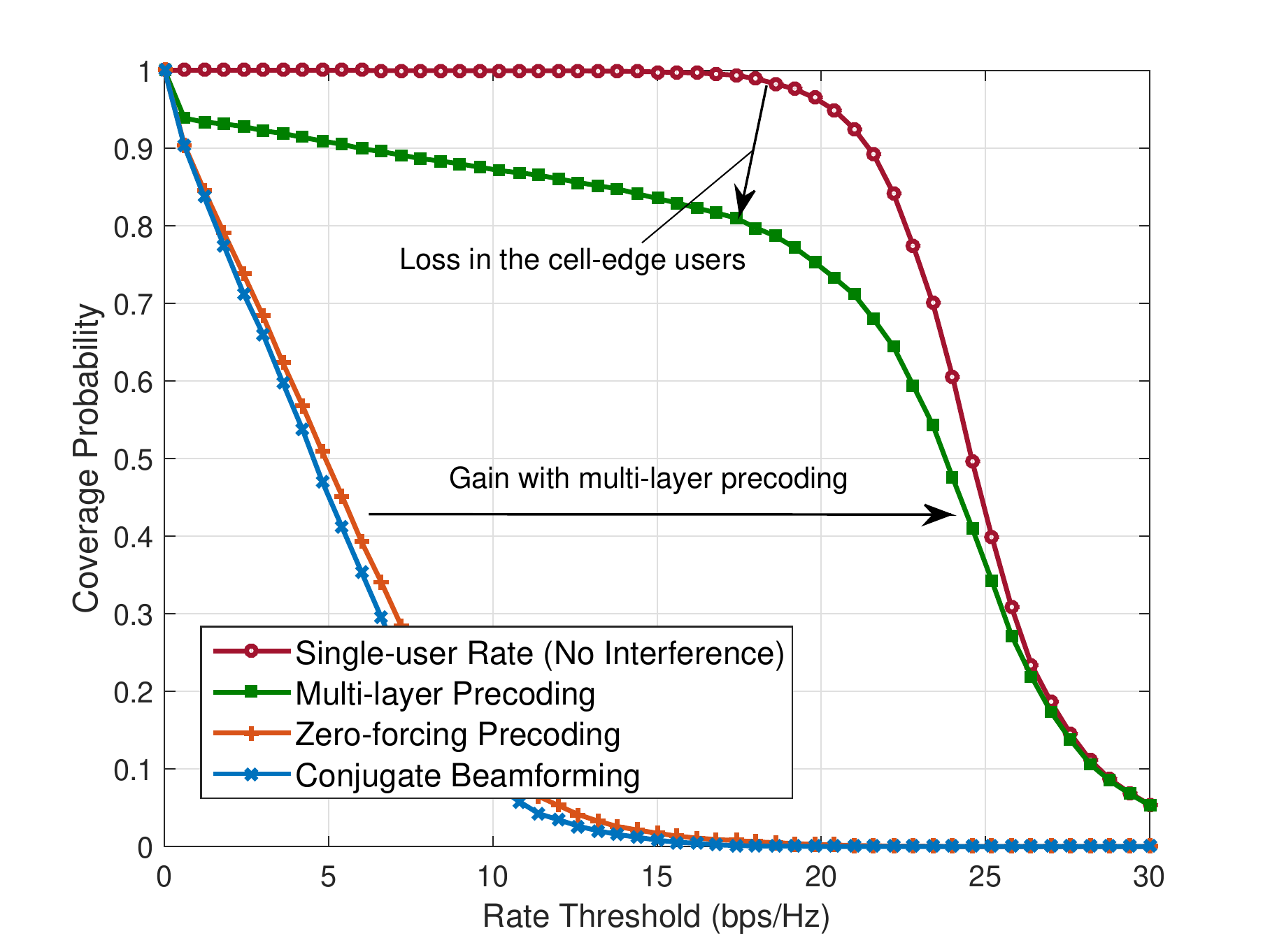}
	\caption{The rate coverage gain of the proposed multi-layer precoding algorithm over conventional conjugate beamforming and zero-forcing is illustrated. This rate coverage is also shown to be close to the single-user case. The BSs are assumed to employ $120 \times 30$ UPAs at heights $H_\rm{BS}=35$m, the cell radius is $r_\rm{cell}=100$m, and the users have single-path channels.}
	\label{fig:Cov_Rank1}
\end{figure}
In \figref{fig:Imact_Rad_Hei}, we evaluate the impact of the BS antenna height and cell radius on the achievable rates. This figure adopts the same system and channel assumptions as in \figref{fig:Ant_Rank1}. In \figref{fig:Imact_Rad_Hei}(a), the achievable rates for multi-layer precoding, single-user, and conjugate beamforming are compared for different antenna heights, assuming cells of radius $200$m. The figure shows that multi-layer precoding approaches single-user rates at higher antenna heights. This is intuitive because forcing the transmission to become in the elevation null-space of the interference may have less impact on the desired signal blockage if higher antennas are employed. Note that the convergence to the single-user rate is expected to happen at lower antenna heights when large arrays are deployed. These achievable rates are again compared in \figref{fig:Imact_Rad_Hei}(b), but for different cell radii. This figure illustrates that a higher cell radius generally leads to less rate because of the higher path loss. Further, the difference between single-user and multi-layer precoding rates increases at higher cell radii. In fact, this is similar to the degradation with smaller antenna heights, i.e., due to the impact of the inter-cell interference avoidance on the desired signal power. For reasonable antenna heights and cell radii, however, the multi-layer precoding still achieves good gain over conventional conjugate beamforming.


\textbf{Rate coverage:} To evaluate the rate coverage of multi-layer precoding, we plot \figref{fig:Cov_Rank1}. The same setup of \figref{fig:Ant_Rank1} is adopted again with cells of radius $100$m, and BSs with $120  \times 30$ UPAs at heights $35$m.  First, the figure shows that multi-layer precoding achieves very close coverage to the single-user case, especially for users not at the cell edge. For example, $\sim 60\%$ of the multi-layer precoding users get the same rate of the single-user case. At the cell edge, some degradation is experienced due to the first precoding layer that filters out-of-cell interference and affects the desired signal power. This loss, though, is expected to decrease as more antennas are employed. The figure also shows significant rate coverage gain over conventional conjugate beamforming and zero-forcing precoding solutions. 

\begin{figure}[t]	
	\centering
	\includegraphics[width=.63\columnwidth]{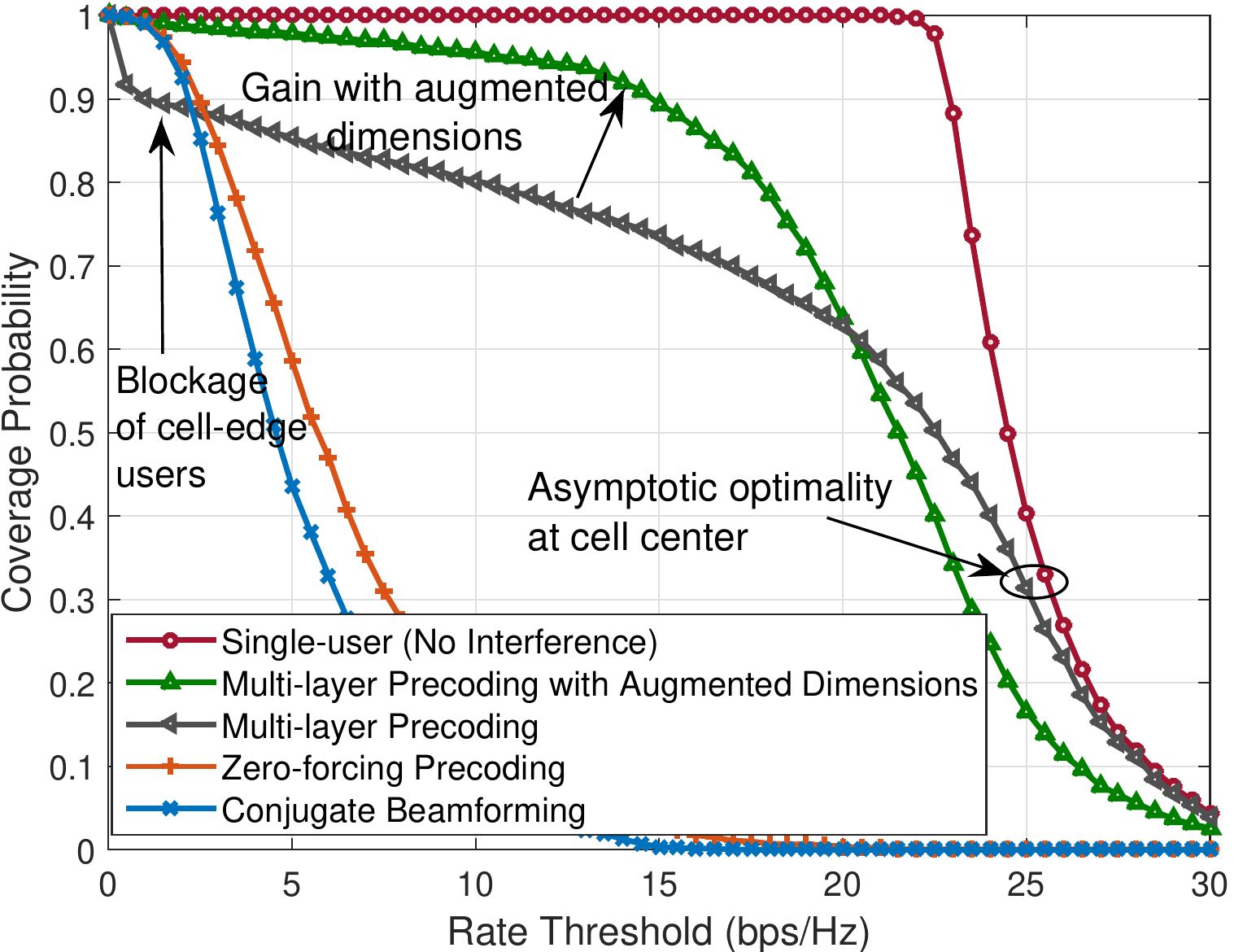}
	\caption{The rate coverage gain of the proposed multi-layer precoding algorithms over conventional conjugate beamforming and zero-forcing is illustrated. This rate coverage is also shown to be close to the single-user case. Further, the modified algorithm with augmented vertical dimensions can overcome the cell-edge blockage. The BSs are assumed to employ $100 \times 40$ UPAs at heights $H_\rm{BS}=35$m, the cell radius is $r_\rm{cell}=100$m. The users have one-ring channel models of azimuth and elevation angular spread $\Delta_\rm{A}=5^\circ, \Delta_\rm{E}=3^\circ$}
	\label{fig:Cov_OneRing_80}
\end{figure}

\begin{figure}[t]	
	\centering
	\includegraphics[width=.63\columnwidth]{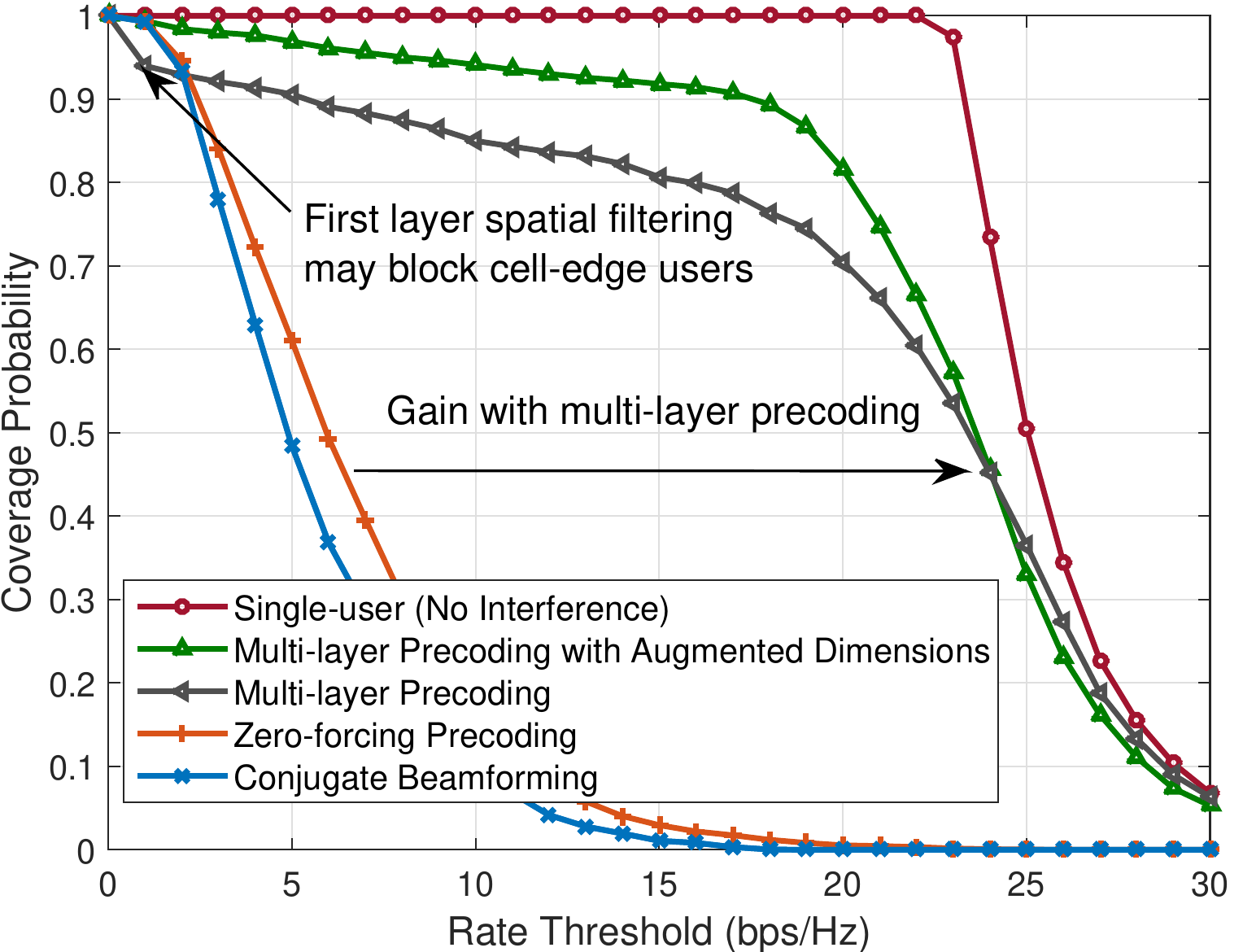}
	\caption{The rate coverage gain of the proposed multi-layer precoding algorithms over conventional conjugate beamforming and zero-forcing is illustrated. This rate coverage is also shown to be close to the single-user case. Further, the modified algorithm with augmented vertical dimensions can overcome the cell-edge blockage. The BSs are assumed to employ $140 \times 40$ UPAs at heights $H_\rm{BS}=35$m, the cell radius is $r_\rm{cell}=100$m.	The users have one-ring channel models of azimuth and elevation angular spread $\Delta_\rm{A}=5^\circ, \Delta_\rm{E}=3^\circ$.}
	\label{fig:Cov_OneRing_140}
\end{figure}

\subsection{Results with One-Ring Channels}
In this section, we adopt a one-ring model for the user channels as described in \eqref{eq:A_Cov}. The azimuth and elevation angles are geometrically determined based on users' locations relative to the BSs, and the angular spread is set to $\Delta_\rm{A}=5^\circ, \Delta_\rm{E}=3^\circ$. Every BS randomly selects $K=20$ users to be served, i.e., no scheduling is done to guarantee the angular separation condition in Theorem \ref{th:CellCenter} and Theorem \ref{th:CellEdge}.

\begin{figure}[t]	
	\centering
	\includegraphics[width=.7\columnwidth]{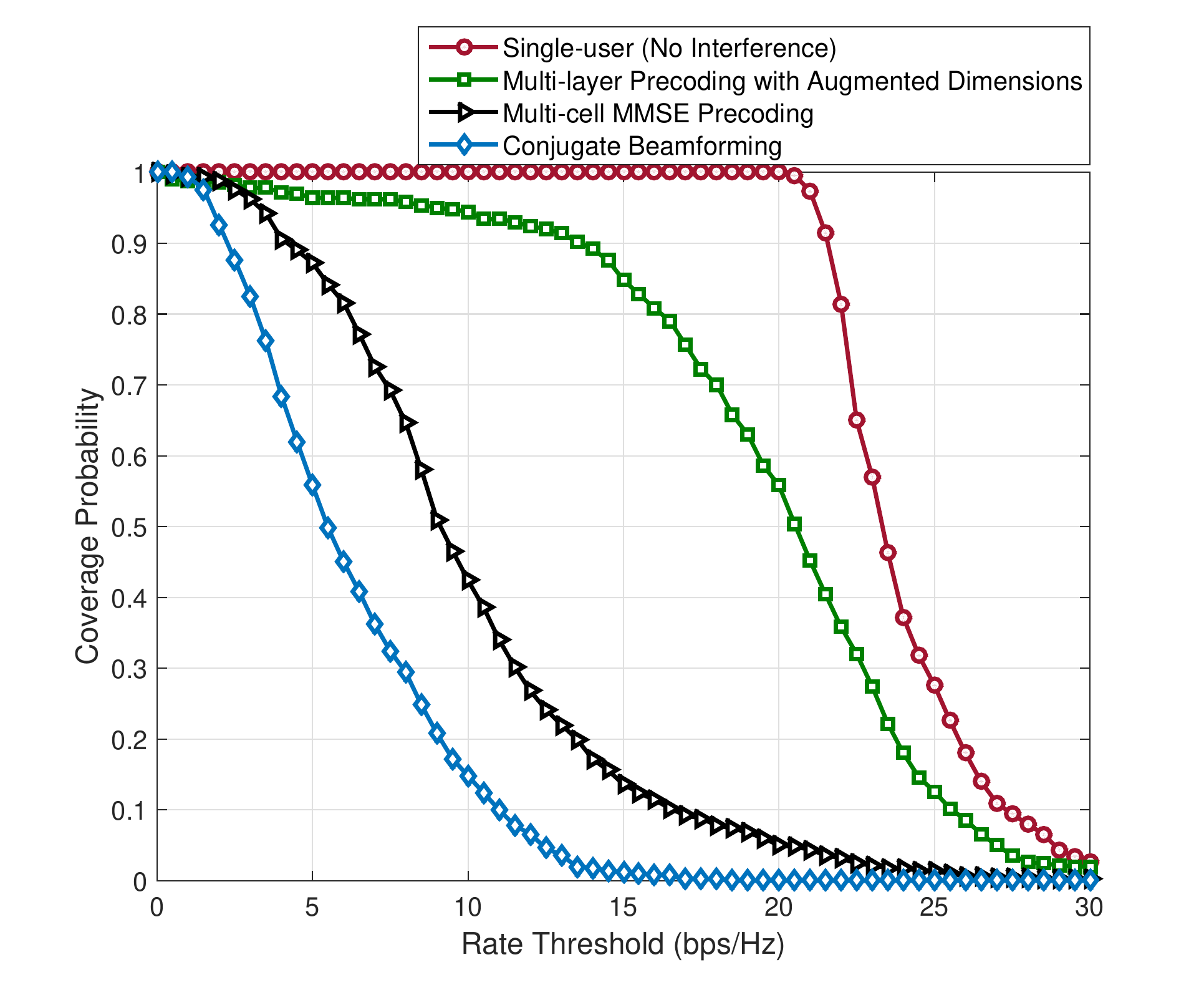}
	\caption{The rate coverage gain of the proposed multi-layer precoding algorithms over conventional single-cell conjugate beamforming and multi-cell MMSE precoding. This rate coverage is also shown to be close to the single-user case. The BSs are assumed to be at heights $H_\rm{BS}=35$m, the cell radius is $r_\rm{cell}=100$m. The users have one-ring channel models of azimuth and elevation angular spread $\Delta_\rm{A}=5^\circ, \Delta_\rm{E}=3^\circ$.}
	\label{fig:Comp_MMSE}
\end{figure}

\textbf{Rate coverage:} 
In \figref{fig:Cov_OneRing_80}-\figref{fig:Comp_MMSE}, we compare the rate coverage of multi-layer precoding, single-user, and conventional conjugate beamforming, for different antenna sizes. We also plot the rate coverage of the multi-layer precoding with augmented vertical dimensions described in \sref{subsec:Aug}, assuming an extended angle $\delta_\rm{E}=2 \Delta_\rm{E}$. This choice makes the maximum no-blockage distance $d_\rm{max}$, defined in Theorem \ref{th:CellCenter}, to be equal to the cell radius. Optimization of this parameter deserves more study in future extensions.
\figref{fig:Cov_OneRing_80} considers the system model in \sref{sec:Model} with $100 \times 40$ BS UPAs and one-ring channel model. First, the figure shows that multi-layer precoding achieves close coverage to the single-user case at the cell center. For the cell edge, though, multi-layer precoding  users experience high blockage, which results from the elevation inter-cell interference avoidance. This can be improved when augmenting vertical subspaces as described in \sref{subsec:Aug}. Different than the multi-layer precoding case, the small degradation at the cell-edge users is due to inter-cell interference, not signal blockage. Further, it is important to note that the cell-center users still achieve the asymptotic optimal rate with the modified algorithm in \sref{subsec:Aug}, i.e., no inter-cell interference  or pilot contamination impact exist. The same behavior is shown again in \figref{fig:Cov_OneRing_140}, when larger array sizes are employed. In this case, though, the cell-edge blockage with multi-layer precoding is less as a better separation between the desired cell and the other cells' users can be achieved.  In the two figures, multi-layer precoding with augmented vertical subspaces is shown to have a good coverage gain over conventional massive MIMO precoding solutions.

In \figref{fig:Comp_MMSE}, we consider the same system and channel models as in \figref{fig:Cov_OneRing_80}, but with $80 \times 20$ UPAs and $K=5$ users to reduce the computational complexity. \figref{fig:Comp_MMSE} compares the rate coverage of the proposed augmented dimension based multi-layer precoding with the single-user rate and the single-cell conjugate beamforming. The figure also plots the rate coverage of the multi-cell MMSE precoding in \cite{Jose2011} that manages the inter-cell interference. As shown in the figure, multi-layer precoding achieves a close performance to single-user rate and good gain over single-cell precoding. \figref{fig:Comp_MMSE} also illustrates that multi-layer precoding achieves a reasonable gain over multi-cell MMSE precoding despite the requirement of less channel knowledge. 


\begin{figure}[t]	
	\centering
	\includegraphics[width=.7\columnwidth]{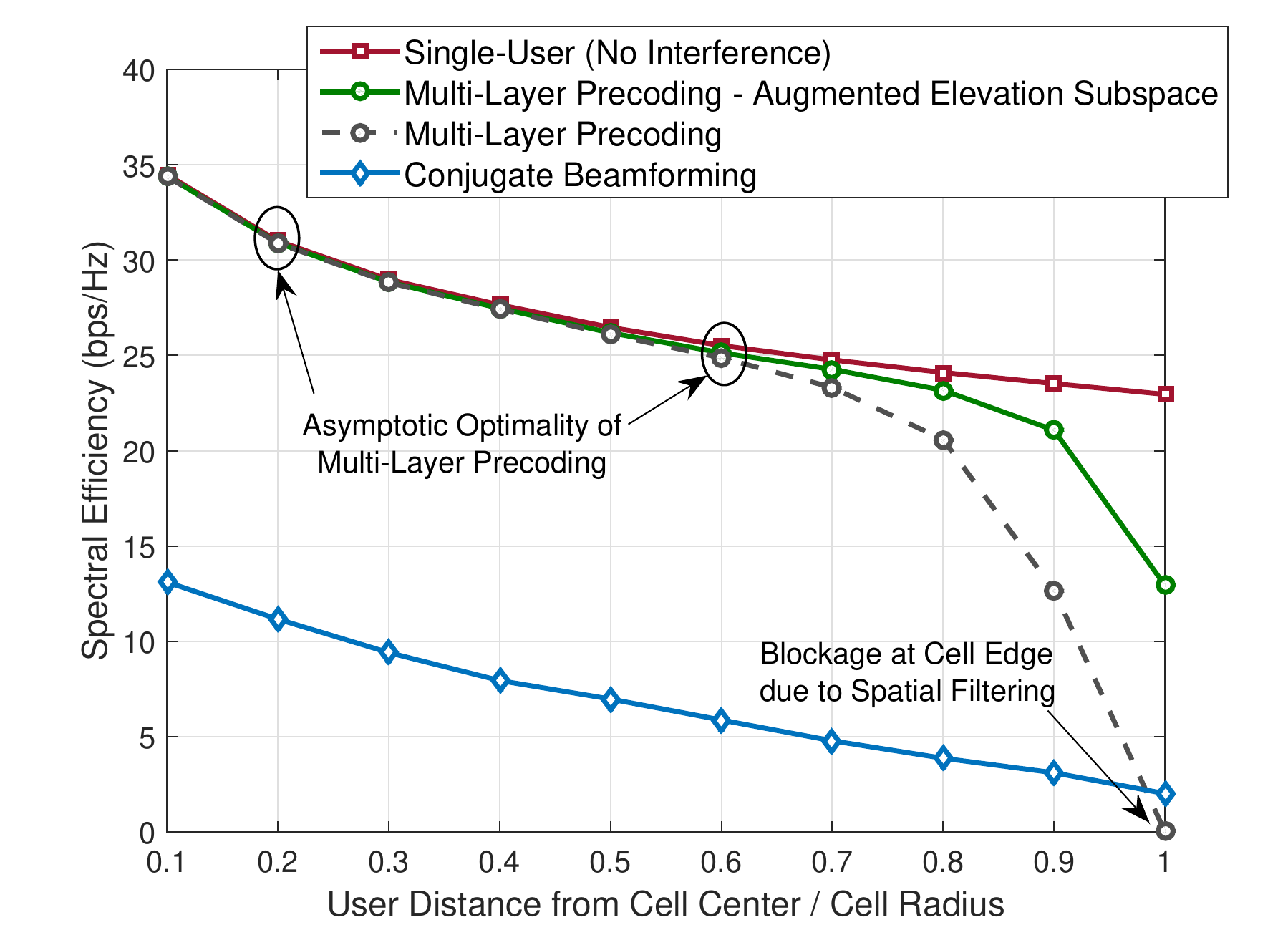}
	\caption{The achievable rates of the proposed multi-layer precoding algorithms are compared to the single-user rate and the rate with conventional conjugate beamforming, for different distances from cell center. The BSs are assumed to employ $120 \times 40$ UPAs at heights $H_\rm{BS}=35$m and the cell radius is $r_\rm{cell}=100$m. The users have one-ring channel models of azimuth and elevation angular spread $\Delta_\rm{A}=5^\circ, \Delta_\rm{E}=3^\circ$.}
	\label{fig:Dist}
\end{figure}

\textbf{Rates at the cell-interior and cell-edge:}
To illustrate the achievable rates for cell-interior and cell-edge users, we plot the achievable rates of multi-layer precoding, single-user, and conventional conjugate beamforming in \figref{fig:Dist}. The rates are plotted versus the user distance to the BS, normalized by the cell radius $r_\rm{cell}=100$m. The figure confirms the asymptotic optimal performance of multi-layer precoding at the cell-interior, given in Theorem \ref{th:CellCenter}. At the cell edge, users experience some blockage that can be fixed with the augmented vertical dimension modification in \sref{subsec:Aug}. Compared to the conventional conjugate beamforming performance, the multi-layer precoding with augmented vertical dimensions has better performance, even at the cell edge.

\section{Conclusion} \label{sec:Conclusion}
In this paper, we proposed a general precoding framework for full-dimensional massive MIMO systems, called multi-layer precoding. We developed a specific design for multi-layer precoding that efficiently manages different kinds of interference, leveraging the large channel characteristics. Using analytical derivations and numerical simulations, we showed that multi-layer precoding can guarantee asymptotically optimal performance for the cell-interior users under the one-ring channel models and for all the users under single-path channels. For the cell-edge users, we proposed a modified multi-layer precoding design that compromises between desired signal power maximization and inter-cell interference avoidance. Results indicated that multi-layer precoding can achieve close performance, in terms of rate and coverage, to the single-user case. Further, results showed that multi-layer precoding achieves clear gains over conventional massive MIMO precoding techniques. For future work, it would be interesting to investigate and optimize the implementation of multi-layer precoding using hybrid analog/digital architectures. It is also important to develop techniques for the channel training and estimation under hybrid architecture hardware constraints.

\appendices 
\section{}\label{app:Ach_Rate}
\begin{proof}[Proof of Lemma \ref{lem:Ach_Rate}]
	To prove the achievable rate in \eqref{eq:Ach_Rate_1}, it is sufficient to prove that the power normalization factor $\left[\bUpsilon\right]_{k,k}$ that satisfies the multi-layer precoding power constraint $\left\|\left[\bF_c^{(1)} \bF_c^{(2)} \bF_c^{(3)}\right]_{:,k}\right\|^2=1$ is given by $\left[\bUpsilon\right]_{k,k}=\sqrt{\left(\left( \bW^*_c {\bF_{c}^{(2)}}^* {\bF_{c}^{(2)}} {\bF_{c}^{(2)}}^* {\bF_{c}^{(2)}}  \bW_c \right)_{k,k}^{-1}\right)^{-1}}$. Using this values of $\left[\bUpsilon\right]_{k,k}$, the multi-layer precoding power constraint can be written as
	\begin{align}
	\left\|\left[\bF_c^{(1)} \bF_c^{(2)} \bF_c^{(3)}\right]_{:,k}\right\|^2 & = \left[\bUpsilon\right]_{:,k}^* \bF_c^{(3)} \bF_c^{(2)} \bF_c^{(1)}\bF_c^{(1)} \bF_c^{(2)} \bF_c^{(3)} \left[\bUpsilon\right]_{:,k} \\
	&\stackrel{(a)}{=} \left[\bUpsilon\right]_{:,k}^* \bF_c^{(3)} \bF_c^{(2)}  \bF_c^{(2)} \bF_c^{(3)} \left[\bUpsilon\right]_{:,k} \\
	&\stackrel{}{=} \left[\bUpsilon\right]_{:,k}^* \left(\overline{\bH}_c^* \overline{\bH}_c\right)^{-1} \overline{\bH}_c^* \bF_c^{(2)}  \bF_c^{(2)} \overline{\bH}_c \left(\overline{\bH}_c^* \overline{\bH}_c\right)^{-1} \left[\bUpsilon\right]_{:,k} \\
	&\stackrel{(b)}{=} \left[\bUpsilon\right]_{:,k}^* \left(\bW_c^* {\bF_c^{(2)}}^* \bF_c^{(2)} {\bF_c^{(2)}}^* \bF_c^{(2)} \bW_c \right)^{-1} \bW_c^* {\bF_c^{(2)}}^* \bF_c^{(2)} \bF_c^{(2)}  \bF_c^{(2)} {\bF_c^{(2)}}^* \bF_c^{(2)} \bW_c \nonumber \\
	& \hspace{10pt} \times \left(\bW_c^*{\bF_c^{(2)}}^* \bF_c^{(2)} {\bF_c^{(2)}}^* \bF_c^{(2)} \bW_c\right)^{-1} \left[\bUpsilon\right]_{:,k} \\	
	&\stackrel{}{=} \left[\bUpsilon\right]_{:,k}^* \left(\bW_c^* {\bF_c^{(2)}}^* \bF_c^{(2)} {\bF_c^{(2)}}^* \bF_c^{(2)} \bW_c \right)^{-1} \bW_c^* {\bF_c^{(2)}}^*  \bF_c^{(2)} {\bF_c^{(2)}}^*   \bF_c^{(2)}  \bW_c \nonumber \\
	& \hspace{10pt} \times \left(\bW_c {\bF_c^{(2)}}^* \bF_c^{(2)} \bW_c^*\right)^{-1} \bW_c^*   {\bF_c^{(2)}}^*      \bF_c^{(2)} {\bF_c^{(2)}}^* \bF_c^{(2)} \bW_c \nonumber \\
	& \hspace{10pt} \times \left(\bW_c^*{\bF_c^{(2)}}^* \bF_c^{(2)} {\bF_c^{(2)}}^* \bF_c^{(2)} \bW_c\right)^{-1} \left[\bUpsilon\right]_{:,k} \\	
	&\stackrel{}{=} \left[\bUpsilon\right]_{:,k}^*  \left(\bW_c {\bF_c^{(2)}}^* \bF_c^{(2)} \bW_c^*\right)^{-1}  \left[\bUpsilon\right]_{:,k} \\
	& \stackrel{(c)}{=} 1,	
	\end{align}
where (a) follows by noting that $\bF_c^{(1)}$ has a semi-unitary structure. The effective channel matrix $\overline{\bH}_c=[\overline{h}_{c 1}, ..., \overline{h}_{c K}]$ with $\overline{h}_{c k}=\left[\bG_{c,(k,1)}, ..., \bG_{c,(k,K)}\right]^* \overline{\bw}_{c c k}, k=1, ..., K$ can also be written as $\overline{\bH}_c={\bF_c^{(2)}}^* \bF_c^{(2)} \bW_c$ with $\bW_c=\bI_{K} \circ \left[\overline{\bw}_{c c 1}, ..., \overline{\bw}_{c c K}\right]$, which leads to (b). Finally, (c) follows by substituting with $\left[\bUpsilon\right]_{k,k}=\sqrt{\left(\left( \bW^*_c {\bF_{c}^{(2)}}^* {\bF_{c}^{(2)}} {\bF_{c}^{(2)}}^* {\bF_{c}^{(2)}}  \bW_c \right)_{k,k}^{-1}\right)^{-1}}$. 	
\end{proof}

\section{}\label{app:CellCenter}
\begin{proof}[Proof of Theorem \ref{th:CellCenter}]
Considering the system and channel models in \sref{sec:Model} and applying the multi-layer precoding algorithm in \sref{sec:Algorithm}, the achievable rate by user $k$ at cell $c$ is given by Lemma \ref{lem:Ach_Rate}
\begin{equation}
R_{ck}=\log_2\left(1+ \frac{\mathsf{SNR}}{\left( \bW^*_c {\bF_{c}^{(2)}}^* {\bF_{c}^{(2)}} {\bF_{c}^{(2)}}^* {\bF_{c}^{(2)}}  \bW_c \right)_{k,k}^{-1}}\right).
\end{equation}	

\noindent If $\bG_{c,(k,m)}=\boldsymbol{0}, \forall m\neq k$ and $\bG_{c,(k,k)}=\bI$, then by noting that the matrix $\bW_c$ has a block diagonal structure and using the matrix inversion lemma \cite{Zhang2006}, we get $\left( \bW^*_c {\bF_{c}^{(2)}}^* {\bF_{c}^{(2)}} {\bF_{c}^{(2)}}^* {\bF_{c}^{(2)}}  \bW_c \right)_{k,k}^{-1}=\left\|\overline{\bw}_{c c k}\right\|^{-2}$. Therefore, to complete the proof, it is sufficient to prove that (i) $\lim_{N_\rm{V}, N_\rm{H}\rightarrow \infty} \bG_{c,(k,m)}=\boldsymbol{0}, \forall m\neq k$ and (ii) $ \lim_{N_\rm{V}, N_\rm{H}\rightarrow \infty} \bG_{c,(k,k)}=\bI$. To do that, we will first present the following useful lemma, which is a modified version of Lemma 3 in \cite{Yin2013}.
\begin{lemma} \label{lem:NullRange}
	Consider a user $k$ at cell $c$ with an azimuth angle $\phi_{c k}$. Adopt the one-ring channel model in \eqref{eq:A_Cov} with an azimuth angular spread $\Delta_\rm{A}$ and correlation matrix $\bR_{c c k}^\rm{A}$. Define the unit-norm azimuth array response vector associated with an azimuth angle $\phi_m$ and elevation angle $\theta_m$ as $\bu_m=\frac{\ba(\phi_m, \theta_m)}{\sqrt{(N_\rm{H})}}$, where $\ba(\phi_m,\theta_m)=\left[1, ..., e^{j k D (N_\rm{H}-1) \sin(\theta_x) \sin(\phi_x)}\right]$. 
	If the angle $\phi_x \notin \left[\phi_{c k}-\Delta_\rm{A}, \phi_{c k}+\Delta_\rm{A}\right]$, then 
\begin{equation}
\bu_x \in \text{Null}\left(\bR_{c c k}^\rm{A}\right), \hspace{10pt} \text{as} \hspace{10pt} N_\rm{H} \rightarrow \infty.
\end{equation}		
\end{lemma}  
\begin{proof}
	First, note that $\left[\bR_{c c k}^\rm{A}\right]_{n_1,n_2}$ in \eqref{eq:A_Cov}, can also be written as 
	\begin{equation}
		\left[\bR_{c c k}^\rm{A}\right]_{n_1,n_2}= \frac{1}{ 2 \Delta_\rm{A}} \int_{-\Delta_\rm{A}}^{\Delta_\rm{A}} \left[\ba\left(\phi_{ck}+\alpha, \theta_{c k} \right) \ba^*\left(\phi_{ck}+\alpha, \theta_{c k} \right)\right]_{n_1,n_2} d \alpha
	\end{equation}
	Then, we have
	\begin{align}
	\bu_m^* \bR \bu_m &= \frac{1}{ 2 \Delta_\rm{A}  N_\mathrm{H}} \int_{-\Delta_\rm{A}}^{\Delta_\rm{A}} \ba^*(\phi_m,\theta_m) \ba\left(\phi_{ck}+\alpha, \theta_{c k} \right) \ba^*\left(\phi_{ck}+\alpha, \theta_{c k} \right) \ba(\phi_m,\theta_m) d \alpha \\
	& = \frac{1}{ 2 \Delta_\rm{A}} \int_{-\Delta_\rm{A}}^{\Delta_\rm{A}}
	\frac{1}{ N_\mathrm{H}} \left| \ba^*(\phi_m,\theta_m) \ba\left(\phi_{ck}+\alpha, \theta_{c k} \right) \right|^2  d \alpha.
	\end{align}
	\noindent Using Lemma 1 in \cite{ElAyach2012a}, we reach 
	\begin{equation}
	\lim_{{N_H}\rightarrow \infty} \bu_m^* \bR \bu_m = 0, \hspace{10pt} \forall \phi_m \notin \left[\phi_{c k}-\Delta_\rm{A}, \phi_{c k}+\Delta_\rm{A}\right].
	\end{equation}
\end{proof}	

\noindent Now, to prove that $\bG_{c,(k,m)}=\overline{\bU}_{c c k}^* \overline{\bU}_{c c r}={{\bU}^{\rm{A}^*}_{c c k}} {\bU}_{c c r}^\rm{A} \otimes {\overline{\bU}^{\rm{E}^*}_{c c k}} \overline{\bU}_{c c r}^\rm{E}=\boldsymbol{0}$, we need to prove that either ${\bU}^{\rm{A}^*}_{c c k} {\bU}_{c c r}^\rm{A}=\boldsymbol{0}$ or ${\overline{\bU}^{\rm{E}^*}_{c c k}} \overline{\bU}_{c c r}^\rm{E}=\boldsymbol{0}$. If $\left|\phi_{c k}-\phi_{c m}\right| \geq 2 \Delta_\rm{A}$, then the columns of $\bU_{c c m}^\rm{A} \in \text{Span}\left\{\frac{\ba(\phi_m)}{\sqrt{(N_\rm{H})}}\left| \phi_{m} \in \left[\phi_{c m}-\Delta_\rm{A}, \phi_{c m}+\Delta_\rm{A}\right]\right. \right\} \subseteq \text{Null}\left(\bR_{c c k}^\rm{A}\right)$ as $N_\mathrm{H} \rightarrow \infty$, which follows from Lemma \ref{lem:NullRange}. This leads to $\lim_{N_\rm{H} \rightarrow \infty} {\bU_{c c k}^\rm{A}}^* \bU_{c c m}^\rm{A} = \boldsymbol{0}$. Similarly, if $\left|\theta_{c k}-\theta_{c m}\right| \geq 2 \Delta_\rm{E}$, then $\lim_{N_\rm{V} \rightarrow \infty} \bU_{c c k}^{\rm{E}^*} \bU_{c c m}^\rm{E} = \boldsymbol{0}$. Further, since $d \leq d_\mathrm{max}$, we have $\left|\theta_{c k}-\theta_{I}\right| \geq 2 \Delta_\rm{E}$, for any elevation angle $\theta_I$ of another cell user. This implies that $\bU_{c c k}^\rm{E} \in \text{Range} \left\{\bU^\rm{NI}_c\right\}$ as $N_\rm{V} \rightarrow \infty$ by Lemma \ref{lem:NullRange}, and $\exists \bA_{c k}$ such that $\bU_{c c k}^\rm{E}=\bU_c^\rm{NI} \bA_{c k}$. For the $\bU_{c c m}$, it can be generally expressed as $\bU_{c c m}=\bU_c^\rm{NI} \bA_{c m} + \bU_c^\rm{I} \bB_{c m}$ for some matrices $\bA_{c m}, \bB_{c m}$ of proper dimensions. As $\lim_{N_\rm{V} \rightarrow \infty} \bU_{c c k}^{\rm{E}^*} \bU_{c c m}^\rm{E} = \boldsymbol{0}$, we have $\lim_{N_\rm{V} \rightarrow \infty} \bA_{c k}^* \bA_{c m} = \boldsymbol{0}$. Then, $\overline{\bU}_{c c k}^{\rm{E}^*} \overline{\bU}_{c c m}^\rm{E}=\bA_{c k}^* \bA_{c m}=\boldsymbol{0}$ as $N_\mathrm{V} \rightarrow \infty$. This completes the proof of the first condition, $\bG_{c,(k,m)}=\boldsymbol{0}$ if $\left|\phi_{c k}-\phi_{c m}\right| \geq 2 \Delta_\rm{A}$ or $\left|\theta_{c k}-\theta_{c m}\right| \geq 2 \Delta_\rm{E}$, $\forall m \neq k$.

To prove that $\lim_{N_\rm{V}, N_\rm{H}\rightarrow \infty} \bG_{c,(k,k)}=\bI$, we need to show that $\overline{\bU}_{c c k}^{\rm{E}^*} \overline{\bU}_{c c k}^\rm{E}=\bI$. Since $\bU_{c c k}^\rm{E}$ can be written as $\bU_{c c k}^\rm{E}=\bU_c^\rm{NI} \bA_{c k}$ when $N_\rm{V} \rightarrow \infty$, then we have $\bU_{c c k}^{\rm{E}^*} \bU_{c c k}^\rm{E}= \bA_{c k}^*  \bA_{c k} = \bI$. This results in $\overline{\bU}_{c c k}^{\rm{E}^*} \overline{\bU}_{c c k}^\rm{E}=\bA_{c k}^*  \bA_{c k} = \bI$ as $N_\rm{V} \rightarrow \infty$, which completes the proof. 
\end{proof}

\section{}\label{app:CellEdge}
\begin{proof}[Proof of Theorem \ref{th:CellEdge}]
Similar to the proof of Theorem \ref{th:CellCenter}, if $\left|\phi_{c k}-\phi_{c m}\right| \geq 2 \Delta_\rm{A}$ or $\left|\theta_{c k}-\theta_{c m}\right| \geq 2 \Delta_\rm{E}$, $\forall m \neq k$, then $\lim_{N_\rm{V}, N_\rm{H} \rightarrow \infty} \bG_{c,(k,m)}=\boldsymbol{0}$. Using the matrix inversion lemma	and leveraging the block diagonal structure of $\bW_c$, we get $\left( \bW^*_c {\bF_{c}^{(2)}}^* {\bF_{c}^{(2)}} {\bF_{c}^{(2)}}^* {\bF_{c}^{(2)}}  \bW_c \right)_{k,k}^{-1} = \left(\overline{\bw}^*_{c c k} \bG_{c (k,k)} \overline{\bw}_{c c k}\right)^{-1}$. Note that since $d>d_\rm{max}$, $\bU_{c c k}$ is not guaranteed to be in $\text{Range}\left(\bUN_c\right)$, and $\overline{\bU}_{c c k}^{\rm{E}^*} \overline{\bU}_{c c k}^\rm{E} \neq \bI$ in general. The achievable rate of user $k$ at cell $c$ can therefore be written as 
\begin{align}
\lim_{N_\rm{V}, N_\rm{H} \rightarrow \infty} R_{c k} &=\log_2\left(1+ \mathsf{SNR} \  \overline{\bw}^*_{c c k} \bG_{c (k,k)} \overline{\bw}_{c c k} \right) \\
& \stackrel{(a)}{\geq} \log_2\left(1+ \mathsf{SNR} \  \left|\overline{\bw}_{c c k} \right\|^2 \sigma^2_\rm{min}\left(\bI \otimes \overline{\bU}_{c c k}^\rm{E}\right) \right) \\
& \stackrel{(b)}{=} \log_2\left(1+ \mathsf{SNR} \  \left|\overline{\bw}_{c c k} \right\|^2 \sigma^2_\rm{min}\left(\overline{\bU}_{c c k}^\rm{E}\right) \right),
\end{align}
where (a) follows by applying the Rayleigh-Ritz theorem \cite{Lutkepohl1997}, and (b) results from the properties of the Kronecker product.
\end{proof}

\linespread{1.25}

\end{document}